\documentclass[aps, prx, reprint, amsmath,amssymb,showpacs,floatfix,longbibliography, onecolumn, superscriptaddress,nofootinbib]{revtex4-2}
\usepackage{times} 
\usepackage{graphicx}
\usepackage{dcolumn}
\usepackage{bm}
\usepackage{bbm}
\usepackage{color}
\usepackage{xcolor}
\usepackage{hyperref}
\usepackage{enumitem}
\usepackage{cancel}

\usepackage{CJKutf8}
\usepackage{tikz-cd}
\usepackage{amsthm}
\usepackage{multirow}
\usepackage{tikz}
\usetikzlibrary{arrows.meta,calc}

\usepackage{comment}
\hypersetup{colorlinks=true,citecolor=blue,linkcolor=blue, urlcolor=blue}
\hypersetup{linktocpage}
\newcolumntype{M}[1]{>{\centering\arraybackslash}m{#1}}
\newcolumntype{N}{@{}m{0pt}@{}}
\usepackage{environ}
\usepackage{MnSymbol}

\bibpunct{[}{]}{,}{n}{}{}

\usepackage{amsfonts,amssymb,amsmath}
\usepackage[T1]{fontenc}
\usepackage{tikz}
\newcommand{\bra}[1]{\langle {#1} |}
\newcommand{\ket}[1]{ | {#1} \rangle}
\let\originalleft\left
\let\originalright\right
\renewcommand{\left}{\mathopen{}\mathclose\bgroup\originalleft}
\renewcommand{\right}{\aftergroup\egroup\originalright}

\newtheorem{definition}{Definition}[section]
\newtheorem{theorem}{Theorem}[section]
\newtheorem{corollary}{Corollary}[theorem]
\newtheorem{lemma}[theorem]{Lemma}

\newtheorem*{remark}{Remark}

\newcommand{\F}{\mathbb F_2}
\newcommand{\G}{\mathcal G}
\newcommand{\St}{\mathcal S}

\newcommand{\Span}{\operatorname{span}}

\newcommand{\de}{\mathrm{de}}
\newcommand{\co}{\mathrm{code}}
\newcommand{\st}{\mathrm{st}}

\newtheorem{proposition}{Proposition}

\NewEnviron{eqs}{%
\begin{equation}\begin{split}
    \BODY
\end{split}\end{equation}
}

\usepackage{amsmath}

\DeclareMathOperator*{\argmin}{arg\,min}

\begin{document}
\begin{CJK*}{UTF8}{gbsn}
\title{
Theory of spacetime quantum fault tolerance}
\author{Yijia Xu (许逸葭)}
\email{yijia@terpmail.umd.edu}

\affiliation{Joint Center for Quantum Information and Computer Science, University of Maryland, College Park,
Maryland 20742, USA}

\author{Yixu Wang}
\email{wangyixu@simis.cn}
\affiliation{Shanghai Institute for Mathematics and Interdisciplinary Sciences (SIMIS), Shanghai 200433, China}

\author{Zi-Wen Liu}
\email{zwliu0@tsinghua.edu.cn}
\affiliation{Yau Mathematical Sciences Center, Tsinghua University, Beijing, 100084, China }

\date{\today}

\begin{abstract}
We present a unified algebraic theory of Clifford spacetime quantum fault tolerance, which underpins universal fault-tolerant quantum computation. 
We construct a spacetime chain complex from a tensor network representation of a circuit that encodes fault propagation, detector syndromes, and logical actions through the associated gauge and stabilizer structure.
We show that this theory (1) unifies various formulations of quantum fault tolerance including spacetime codes, Gottesman's gadget framework, fault complexes, and ZX calculus; (2) connects spacetime distances to the gadget-level correctness properties in Gottesman's formalism; (3) naturally accommodates a range of fault models, in particular phenomenological and circuit-level models, bringing these often separately formulated descriptions into a common algebraic framework through different tensor decompositions and fault assignments; and (4) recovers many known fault tolerance constructions and results as special cases, including repeated syndrome extraction, correlated detectors for transversal and fold-transversal gates, Steane- and Knill-type syndrome extraction gadgets, and single-shot error correction.
\end{abstract}

\maketitle
\end{CJK*}



\section{Introduction}

Fault-tolerant quantum computation requires quantum information to remain protected as it is stored, manipulated, and measured using noisy circuit components, making it an inherently spacetime dynamical problem. Extensive research on fault-tolerant quantum computation has focused on the construction and analysis of quantum error-correcting codes. However, static code properties alone are not sufficient: the circuits implementing logical operations must also be robust to physical faults and their propagation throughout the computation. Fault tolerance has been understood from several complementary perspectives. For example, the existence of positive noise thresholds provides a strong asymptotic guarantee, but establishing such a threshold requires probabilistic analysis in addition to structural properties of the circuit. At the level of individual circuits or gadgets, algebraic quantities such as fault distance are often used to characterize robustness to faults and to inform threshold analyses.

Several theoretical formalisms have been developed to characterize fault tolerance in quantum circuits. Gottesman's gadget framework~\cite{gottesman2024surviving} builds circuit-level fault tolerance guarantees from properties of individual gadgets and their composition. More recently, spacetime codes~\cite{bacon2017sparse,delfosse2023spacetime,pesah2025fault,lee2026gauging},
circuit detector formalisms~\cite{mcewen2023relaxing,derks2025designing},
measurement-based fault tolerance and fault complexes~\cite{logicalblocks,faultcomplex},
and ZX calculus~\cite{bombin2024unifying,rodatz2025fault,rusch2025completeness}, circuit Tanner graph \cite{li2025ldpc} have provided frameworks for characterizing different aspects of fault tolerance in quantum circuits. Building on these perspectives, substantial theoretical and experimental progress has been made toward reducing the overhead of fault-tolerant quantum computation, including advances in syndrome extraction protocols, logical gate implementations, and decoding methods~\cite{mcewen2023relaxing,eickbusch2025demonstration,shaw2025morphing,shaw2026optimising,xu2025batched,cowtan2025fast,williamson2026low,correlated_decoding,poor2026ultra}. Nevertheless, a general understanding of when the structural properties
captured by these frameworks suffice to guarantee gadget fault tolerance is still lacking.
At the same time, although these frameworks are often expected to describe closely related structures, their relationships have yet to be systematically clarified.  Establishing concrete translations between them is important both conceptually and practically: it clarifies when constructions, distance bounds, and decoding techniques developed in one framework can be transferred to another, while also providing systematic tools for constructing and verifying fault-tolerant circuits.

In this work,  by combining circuit-to-tensor mappings and homological chain complexes, we develop a unified algebraic theory for spacetime fault tolerance in Clifford circuits that form the skeleton of universal fault-tolerant quantum computation. The microscopic description of our formalism is given by the circuit-to-tensor correspondence through which local tensor relations and circuit boundary conditions define a homological chain complex describing a spacetime subsystem code. The stabilizers associated with individual tensors serve as gauge operators of this spacetime subsystem code. This representation connects spacetime fault tolerance to static error-correcting codes through a common algebraic structure: dynamical elements, including fault propagation, detector constraints, and logical actions, are naturally encoded in this structure and can be directly recovered from it.  In particular, we derive sufficient and necessary criteria for spacetime Clifford circuits, with and without feedback, to satisfy the corresponding gadget correctness conditions. These results bridge circuit-level fault distance and gadget-level correctness, thereby laying the foundation for full fault tolerance analysis of broad classes of quantum circuits and protocols.

Our algebraic framework also offers a unified perspective on the aforementioned fault-tolerance formalisms, clarifying the relationships among them and thereby providing a systematic scheme for constructing and analyzing fault-tolerant circuits.

We apply our formalism to examine a variety of established fault tolerance constructions and results, including repeated syndrome extraction, correlated detectors for transversal and fold-transversal gates, Steane- and Knill-type syndrome extraction gadgets, and single-shot error correction. These examples demonstrate the versatility of our algebraic framework for analyzing fault-tolerant Clifford circuits.

\section{Unified formalism of spacetime quantum fault tolerance}

\subsection{Representing circuit elements as tensors}

The microscopic description of our framework is formulated in terms of tensors. An $l$-leg tensor $T$ is an element of the tensor product of $l$ finite-dimensional vector spaces, $T\in V_1\otimes\cdots\otimes V_l$. The linear nature of quantum mechanics makes tensors a particularly suitable language for describing quantum processes. After choosing bases to identify vector spaces with their duals, the natural isomorphism $V\otimes W^*\simeq\operatorname{Hom}(W,V)$ allows an $l$-leg tensor to be interpreted as a linear map from $V_1\otimes\cdots\otimes V_{l_i}$ to $V_{l_i+1}\otimes\cdots\otimes V_l$, where $l_i+l_o=l$. At the algebraic level, the same tensor can therefore be interpreted with different choices of input and output legs. Quantum computation is commonly described by the circuit model, whose basic ingredients include quantum gates, state initialization, measurements, and feedback based on measurement outcomes. All these ingredients admit tensor representations. For example, an $n$-qubit pure state is an $n$-leg tensor, while an $n$-qubit unitary gate is a $2n$-leg tensor. The tensor description can also accommodate gate faults and correlated noise models.

In universal fault-tolerant quantum computation, Clifford circuits implement logical Clifford operations and underpin central procedures including stabilizer state preparation, syndrome extraction, gate teleportation, and code switching. The stabilizer channel formalism provides a natural setting for analyzing such encoded processes. The Clifford circuits considered here consist of the following ingredients: (1) single-qubit initialization and measurement in the Pauli $X$ or $Z$ basis; (2) Clifford gates generated by $\{\mathrm{CNOT},\mathrm{H},\mathrm{S}\}$; and (3) Pauli gates conditioned on previous measurement outcomes. Developing a systematic understanding of the fault tolerance properties of these circuits is therefore fundamental to fault-tolerant quantum computation and is the central aim of this work.


A useful feature of Clifford circuits is that many of their components can be characterized by an associated gauge group of Pauli operators. For initial states which are Pauli stabilizer states, it is well-known that they can be equivalently described by a set of Pauli stabilizers. The following lemmas make precise how these gauge operators determine circuit components of Clifford unitary gates and Pauli measurements.

\begin{lemma}[Clifford unitary]
Let $U:\mathbb{C}^{2^n}\to\mathbb{C}^{2^n}$ be an $n$-qubit Clifford unitary operator. Then the Choi state $\ket{\mathcal{U}}\equiv (\mathbb{I} \otimes U) \ket{\Omega}$, where $\ket{\Omega}\equiv 2^{-n/2}\sum_{\mathbf{x}\in \mathbb{F}_2^n} |\mathbf{x\rangle } \otimes | \mathbf{x\rangle}$, is a $2n$-qubit Pauli stabilizer state, with gauge generators

\begin{eqs}
    \mathcal{R}=\langle X(j) \otimes U X(j) U^\dagger,  Z(j) \otimes U Z(j) U^\dagger:~ \forall j=1,\ldots,n \rangle.
\end{eqs}

\end{lemma}

\begin{lemma}[Pauli operator measurement]
Let $M:\mathbb{C}^{2^n}\to\mathbb{C}^{2^n}$ be an $n$-qubit Pauli observable to be measured. Then the density matrix of each post-measurement state $\ket{\psi_M}\bra{\psi_M}$ is stabilized by $M$, $\ket{\psi_M}\bra{\psi_M}=M\ket{\psi_M}\bra{\psi_M}M^\dagger$.
\end{lemma}

\subsection{Circuit as spacetime tensor network}

A circuit is composed of different components. In terms of tensors, the composition corresponds to contracting tensors into a spacetime tensor network. The contracted legs of the tensors provide information about the spatial location and the sequential order of circuit components.

In terms of the gauge operators of each tensor, the contraction is described as follows. For two tensors with gauge group $\mathcal{R}_1, \mathcal{R}_2$, suppose we want to contract the $i$-th qubit of the first tensor with the $j$-th qubit of the second tensor, we effectively project the contracted two qubits into the Bell state which is stabilized by $X(i,1) X(j,2), Z(i,1) Z(j,2)$. This process can be regarded as a subsystem code whose gauge group is
\begin{eqs}
    \mathcal{R}=\langle \mathcal{R}_1, \mathcal{R}_2 ,  X(i,1) X(j,2), Z(i,1) Z(j,2) \rangle.
\end{eqs}

The single contraction can be directly generalized to the contraction of multiple legs and multiple tensors. The resulting tensor is the network that represents the full circuit. For a general spacetime tensor network, we can describe it by the gauge group 
\begin{eqs}
    \mathcal{G}=\langle \mathcal{S}_{\mathrm{in}},  \mathcal{R}, \mathcal{B} , \mathcal{M} \rangle.
\end{eqs}
Here, for each input initial state or ancilla initialization, we put its stabilizers in the set $\mathcal{S}_{\mathrm{in}}$ and regard them as gauge generators. \textcolor{black}{Note that we do not fix its gauge at this stage, because we would like to describe the fault tolerance property of a family of input states.} For each Clifford gate, we put its gauge operators in the set $\mathcal{R}$. For each contracted pair of legs, say leg $a$ of gate $k$ and leg $b$ of gate $m$,  we put the gauge generators $X(a,k) X(b,m), Z(a,k) Z(b,m)$ in the set $\mathcal{B}$ and call them the bond gauge. For each Pauli measurement, we put the measurement operator in the set $\mathcal{M}$ and call them the measurement gauge.

We call the starting ends of all the circuit wires the input ports, and the finishing ends the output ports of a circuit. Input and output ports carry the quantum boundary degrees of freedom, on which the input gauge group $\mathcal{S}_{\mathrm{in}}$ and the measurement gauge group $\mathcal{M}$ act, respectively.

In general, we describe the circuit-to-tensor mapping as follows:
\begin{enumerate}
    \item \textbf{Time slicing and identity padding.}
Divide the circuit into discrete time steps. Each time step contains exactly one layer of Clifford circuits or state initialization. If a qubit is not acted on by a non-identity gate in a given time step, insert an identity gate acting on this qubit at this step.

\item \textbf{Input boundary.}
    Represent each input port $q$ by a single spacetime-code qubit. The collection of input ports is $Q_{\mathrm{input}}$. If an input state is inserted at time $t$ at ports $q$, put its stabilizer $S_o^{(t)}(q)$ in the input gauge group \(\mathcal S_{\mathrm{in}}\equiv  \langle S_o^{(t)}(q): \forall q \in Q_{\mathrm{input}}\rangle\) which includes ancilla initialization. Equivalently, the spacetime circuit may be viewed as beginning at the output boundary of the preceding gadget.

\item \textbf{Gates or gadgets.}
    Replace each gate or gadget with its Choi tensor. For a gate or gadget $k$ inserted at time $t$, acting on a set of qubits $1,\ldots,n$, its gauge group is 
    \begin{eqs}
    \mathcal{R}_k=\langle X_i^{(t)}(j) \otimes U X_o^{(t)}(j) U^\dagger,  Z_i^{(t)}(j) \otimes U Z_o^{(t)}(j) U^\dagger:~ \forall j=1,\ldots,n  \rangle.
    \end{eqs}
    If the circuit contains $L$ elementary gadgets with local gauge groups $\mathcal R_1,\ldots,\mathcal R_L$, define
    \begin{eqs}
    \mathcal{R}\equiv\langle  \mathcal{R}_1,\ldots, \mathcal{R}_L \rangle.
    \end{eqs}

\item \textbf{Internal circuit connections.}
    For every circuit wire $e$ connecting two consecutive time steps $t$ and $t+1$, retain the two associated tensor-leg coordinates as physical qubits of the spacetime code and introduce the bond gauge subgroup 
    \begin{eqs}
        \mathcal{B}_e=\langle X^{(t)}_{o}(e) X^{(t+1)}_{i}(e), Z^{(t)}_{o}(e) Z^{(t+1)}_{i}(e)\rangle. 
    \end{eqs}
    The  total bond gauge group is $
        \mathcal{B}\equiv \langle \mathcal{B}_e: \forall e\rangle.$

    \item \textbf{Output boundary.}
    The output boundary $Q_{\mathrm{output}}$ is formed by the output ports of all the circuit wires.

    \item \textbf{Measurements.}
    For each Pauli measurement $P$ executed at time $t$ at place $a$, put $P_o^{(t)}(a)$ in the measurement gauge set.
    \begin{eqs}
        \mathcal{M}\equiv\langle P_o^{(t)}(a):  a\in \Lambda_{\mathrm{meas}}\subseteq Q_{\mathrm{output}} \rangle, 
    \end{eqs}
    where $\Lambda_{\mathrm{meas}}$ denote the measurement positions. Different measurements are assumed to be supported on disjoint positions. 
    The eigenvalue of $P_o^{(t)}(a)$ records the corresponding measurement outcome.

\end{enumerate}

To facilitate further discussions, we describe a quantum circuit as follows. We label all the qubits that are ever involved in a given circuit as $q=1,\ldots, n_{\mathrm{tot}}$. At each time step, qubits can be added to the circuit, either in a known state, or in an unknown logical state of a stabilizer code. Each qubit enters the circuit at time $t_{i}(q)$. Afterwards, they exit the circuit at a later time $t_{o}(q)\geq t_{i}(q)$. When exiting, a qubit is either preserved as part of the circuit output, or destructively measured.\footnote{We understand a non-destructive measurement as the measured data qubits coupled to some ancilla and the ancilla being destructively measured. In this sense, every measurement is destructive.} Therefore, each physical qubit introduces spacetime locations $q_{o}^{(t_{i}(q))},q_{i}^{(t_{i}(q)+1)},\ldots, q_{i}^{(t_{o}(q))}, q_{o}^{(t_{o}(q))}$, in total $2(t_{o}(q)-t_{i}(q))+1$ locations. The oddness of the number of locations for each qubit is crucial to produce the correct commutation relations between spacetime code stabilizers and logical operators, and was first discussed in Ref.~\cite{bacon2017sparse}. The collection of all the locations is denoted as $\mathcal{L}\equiv\{q_i^{(t)}, q_o^{(t)}\}$. The collection of locations $Q_{\mathrm{input}}\equiv\{q_{o}^{(t_{i}(q))}\}$ denotes the input boundary of a circuit and $Q_{\mathrm{output}}\equiv\{ q_{o}^{(t_{o}(q))}\}$ denotes the output boundary respectively. The gauge group $\mathcal S_{\mathrm{in}}$ is supported exclusively on the input boundary, while the measurement gauge group $\mathcal{M}$ is supported exclusively on the output boundary.

Up till now, the gauge group description helps us to translate a quantum circuit into a tensor network. It is natural to use the same gauge group to define a quantum code, which we call the spacetime subsystem code.
\begin{definition}[Spacetime subsystem code] The spacetime subsystem code is defined on the qubits on the spacetime locations $\mathcal{L}$ induced by the circuit-to-tensor mapping. The gauge group of the spacetime subsystem code is the gauge group of the spacetime tensors, 
\begin{eqs}
    \mathcal{G}=\langle \mathcal{S}_{\mathrm{in}},  \mathcal{R}, \mathcal{B} , \mathcal{M} \rangle.
\end{eqs}
The stabilizer group of the spacetime subsystem code is the center of its gauge group,
\begin{eqs}
    \mathcal S\equiv\mathrm{Cen}(\mathcal G)=\mathcal{G}^\perp\cap \mathcal G.
\end{eqs}   
\end{definition}

A crucial difference in our treatment of the tensor network and previous approaches, such as that in Ref.~\cite{Charles2022Lego}, is that the latter assign physical degrees of freedom only to dangling legs. In the present framework of circuit-to-tensor mapping, each pair of contracted legs and each measured leg in the tensor network still represent a fault position, respectively. Therefore, they carry physical degrees of freedom in the spacetime subsystem code. Specifically, on a pair of contracted legs $e$, the bond gauge group is $\mathcal{B}_e=\langle X_{L,e} X_{R,e},~Z_{L,e} Z_{R,e}\rangle,$ where $L,R$ denote the two qubits in the same bond. Modulo a phase-free 2-qubit Pauli group $\mathcal{P}_2$ by $\mathcal{B}_e$ gives a phase-free single-qubit Pauli group $\mathcal{P}_2/ \mathcal{B}_e \cong \mathcal{P}_1$. Therefore, every two-qubit Pauli fault on a contracted leg is either a bond gauge operator or is gauge equivalent to a Pauli fault acting on only one of the two qubits. In this sense, the two physical qubits of the spacetime code on each contracted leg represent one fault position in the original circuit.

We consider the edge-flip error model for the spacetime subsystem code, where a single-qubit Pauli error channel is applied to each spacetime Pauli location in $\mathcal{L}$. Each individual tensor represents a gadget whose internal operations are ideal, with faults acting only on its external legs. This model is powerful and flexible enough to describe various error models in the literature. This is because we can choose the tensor composition or decomposition rules to assign faultless or faulty locations. This allows different noise models, including phenomenological and circuit-level models, to be represented within the same framework.

For example, the phenomenological error model of syndrome extraction includes Pauli errors on the data qubits and flips of measurement outcomes, while treating syndrome extraction itself as ideal. This error model corresponds to a tensor description in which each syndrome extraction gadget is represented by a composite ideal tensor, with data errors and measurement-outcome errors assigned to the appropriate external legs. The effects of gate faults within the syndrome extraction gadget are therefore coarse-grained into effective Pauli errors on data qubits and flips of measurement outcomes. We analyze this example in Section~\ref{sec:foliation}. A finer decomposition of the spacetime tensor into a network of tensors of degree at most three yields a version of the circuit-level noise model. In this description, faults can occur within the larger gadgets that the coarse-grained model treated as ideal.

\begin{figure}[h]
    \centering
    \includegraphics[width=0.7\textwidth]{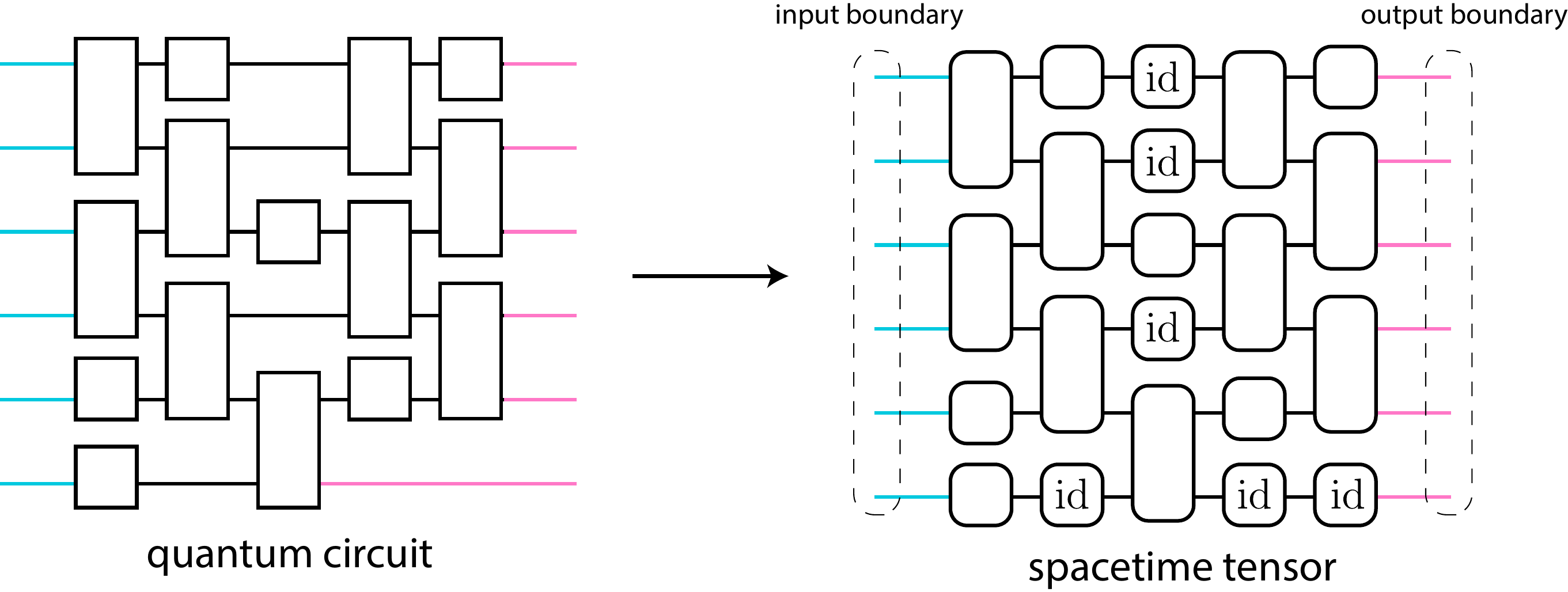}
    \caption{Mapping a quantum circuit to a spacetime tensor. Left: Each rectangle represents a nonidentity unitary gate. Right: The circuit is decomposed into elementary gate tensors depicted in rounded rectangles, with identity tensors inserted on idle qubits to obtain a uniform spacetime structure. The resulting tensor network defines a subsystem code with gauge group $\langle \mathcal{R}, \mathcal{B} ,\mathcal{M},\mathcal{S}_{\mathrm{in}}\rangle$. Each contracted leg represents a fault location in the circuit and is associated with two physical qubits. Each input or output port represents a single physical qubit. The dashed boxes indicate the input and output boundaries. }
    \label{fig:circuit_to_code}
\end{figure}

\subsection{Structural properties of the spacetime subsystem code}
Following the circuit-to-tensor mapping, we can establish universal properties of the spacetime subsystem code. We consider a spacetime circuit implementing the unitary evolution $ U_{\tau} U_{\tau-1}\cdots U_1$. Each $U_i$ acts on the qubits that are present in the circuit at time step $i$.  Define the accumulated evolution up to time $t$ as 
    \begin{eqs}
        \mathcal{U}_{t}\equiv U_t U_{t-1} \cdots U_1,\quad \mathcal{U}_{0}\equiv\mathrm{id}.
    \end{eqs}

The following lemma describes the form of the spacetime code stabilizer. 

\begin{lemma}\label{lemma:spacetimestabilizer}
The stabilizer group of the spacetime subsystem code is generated
by operators of the form
 \begin{eqs}\label{eq:spacetimestabilizer}
\bigotimes_{q=1}^{n_{\mathrm{tot}}}\left( P_{o}^{(t_{i}(q))}(e_q) \otimes\bigotimes_{t=t_{i}(q)+1}^{t_{o}(q)} \left(\mathcal{U}_{t-1} P_{i}^{(t)}(e_q)\mathcal{U}_{t-1}^\dagger
\otimes \mathcal{U}_{t} P_{o}^{(t)}(e_q)\mathcal{U}_{t}^\dagger\right)\right),
\end{eqs}
where $P_o^{(t)}(e_q)$ and $P_i^{(t)}(e_q)$ denote a fixed Pauli operator $P$ supported on the specified locations. Further, it satisfies the commutation relations at the input and output boundary,
\begin{eqs}
[\bigotimes_{q=1}^{n_{\mathrm{tot}}} P_{o}^{(t_{i}(q))}(e_q),s]=0,~\forall s\in \mathcal{S}_{\mathrm{in}}, \qquad [\bigotimes_{q=1}^{n_{\mathrm{tot}}}\mathcal{U}_{t} P_{o}^{(t_{o}(q))}(e_q)\mathcal{U}_{t}^\dagger,m]=0,~\forall m\in \mathcal{M},
\end{eqs}
and either $\bigotimes_{q=1}^{n_{\mathrm{tot}}} P_{o}^{(t_{i}(q))}(e_q)\in \mathcal{S}_{\mathrm{in}}$, or $\bigotimes_{q=1}^{n_{\mathrm{tot}}}\mathcal{U}_{t} P_{o}^{(t_{o}(q))}(e_q)\mathcal{U}_{t}^\dagger\in \mathcal{M}$, or both.

\end{lemma}

\begin{proof}
The proof follows from solving the requirement $S=\mathcal{G}^\perp\cap \mathcal G$. For a spacetime stabilizer to commute with each gauge element in the bond gauge group, the two Pauli operators in the stabilizer supported on each adjacent pair of locations $q_o^{(t)}$ and $q_i^{(t+1)}$ must be identical. Similarly, to commute with each gate gauge generator, at each adjacent pair of locations $q_i^{(t)}$ and $q_o^{(t)}$, the two Pauli operators in the stabilizer supported on these locations have to take the form $P_i^{(t)}\otimes U_t P_o^{(t)} U_t^\dagger$ for a common choice of Pauli $P$. Note that $\mathcal{U}_{t_i(q)} P_{i}^{(t_i(q)+1)}(e_q)\mathcal{U}_{t_i(q)}^\dagger=P_{i}^{(t_i(q)+1)}(e_q)$, following the rules of the circuit-to-tensor map.  These two sets of constraints already give the form of Eq. \eqref{eq:spacetimestabilizer}. The two sets of commutation constraints arise from the requirement that spacetime stabilizers commute with $\mathcal{S}_{\mathrm{in}}$ and $\mathcal{M}$. The final two requirements ensure that a spacetime stabilizer is in the gauge group. Indeed, if $\bigotimes_{q=1}^{n_{\mathrm{tot}}} P_{o}^{(t_{i}(q))}(e_q)\in \mathcal{S}_{\mathrm{in}}$, it can be generated by the elements in $\mathcal{S}_{\mathrm{in}}\cup \mathcal{R}$. If $\bigotimes_{q=1}^{n_{\mathrm{tot}}}\mathcal{U}_{t} P_{o}^{(t_{o}(q))}(e_q)\mathcal{U}_{t}^\dagger\in \mathcal{M}$, it can be generated by the elements in $\mathcal{B}\cup\mathcal{M}$. 
\end{proof}

The following lemma clarifies the correspondence between the logical contents of the spacetime code and that of the original input code. As expected, they should be identical if the circuit does not implement any logical measurement.

\begin{definition}[Stabilizer tube, logical measurement and detector]\label{def:detector_stabilizer_tube}
Following the notations in Lemma \ref{lemma:spacetimestabilizer}, for a spacetime stabilizer taking the form of Eq. \eqref{eq:spacetimestabilizer}, if it satisfies $\bigotimes_{q=1}^{n_{\mathrm{tot}}} P_{o}^{(t_{i}(q))}(e_q)\in \mathcal{S}_{\mathrm{in}}$, but $\bigotimes_{q=1}^{n_{\mathrm{tot}}}\mathcal{U}_{t} P_{o}^{(t_{o}(q))}(e_q)\mathcal{U}_{t}^\dagger$ is not supported on $\Lambda_{\mathrm{meas}}$, then it is called a stabilizer tube. If it satisfies $\bigotimes_{q=1}^{n_{\mathrm{tot}}}\mathcal{U}_{t} P_{o}^{(t_{o}(q))}(e_q)\mathcal{U}_{t}^\dagger\in \mathcal{M}$, but $\bigotimes_{q=1}^{n_{\mathrm{tot}}} P_{o}^{(t_{i}(q))}(e_q)\notin \mathcal{S}_{\mathrm{in}}$, then it is called a logical measurement. If it satisfies both $\bigotimes_{q=1}^{n_{\mathrm{tot}}} P_{o}^{(t_{i}(q))}(e_q)\in \mathcal{S}_{\mathrm{in}}$, and  $\bigotimes_{q=1}^{n_{\mathrm{tot}}}\mathcal{U}_{t} P_{o}^{(t_{o}(q))}(e_q)\mathcal{U}_{t}^\dagger\in \mathcal{M}$, then it is called a detector. 
\end{definition}

\begin{lemma}[Logical content of spacetime code]\label{lemma:spacetimelogical}
Let a circuit have $n_{\mathrm{tot}}$ wires. Let the rank of the input gauge group $\mathcal{S}_{\mathrm{in}}$ be $ n_{\mathrm{tot}}-k$, and the rank of the measurement gauge group be $n_{\mathcal{M}}$. If the spacetime code stabilizers are either stabilizer tubes or detectors, then the spacetime subsystem code encodes $k$ logical qubits.
\end{lemma}

\begin{proof}
Note that the ansatz in Eq. \eqref{eq:spacetimestabilizer} defines a group isomorphism $C$ from the Pauli group $\mathcal{P}^{\otimes\mathcal{L}}$ supported on $Q_{\mathrm{input}}$ to the Pauli group supported on $Q_{\mathrm{output}}$. Specifically, $C: \bigotimes_{q=1}^{n_{\mathrm{tot}}} P_{o}^{(t_{i}(q))}(e_q)\mapsto\bigotimes_{q=1}^{n_{\mathrm{tot}}}\mathcal{U}_{t} P_{o}^{(t_{o}(q))}(e_q)\mathcal{U}_{t}^\dagger$. Therefore, we can consider the inverse isomorphism acting on $\mathcal{M}$, $C^{-1}(\mathcal{M})$. If the spacetime code stabilizers are either stabilizer tubes or detectors, then if $m\in C^{-1}(\mathcal{M})$ commutes with every element $s\in \mathcal{S}_{\mathrm{in}}$, we have $m\in \mathcal{S}_{\mathrm{in}}$, and vice versa.  If otherwise, plugging $m$ into the ansatz in Eq.\eqref{eq:spacetimestabilizer} yields a valid spacetime stabilizer that is neither a stabilizer tube nor a detector, contradicting the assumption. 
Now consider the number of generators of spacetime stabilizer tubes. We show that, under the assumption of the lemma, it is equivalent to $n_{\mathrm{tot}}-k-n_{\mathcal{M}}$. This is because 
\begin{eqs}
    \text{\# of stabilizer tubes}=& \text{\# of generators of}~ \mathcal{S}_{\mathrm{in}}~\text{that commute with } C^{-1}(\mathcal{M})~ \text{but are not in}~ C^{-1}(\mathcal{M})\\
    =&n_{\mathrm{tot}}-k-\text{\# of generators of $\mathcal{S}_{\mathrm{in}}$ that anticommute with some elements in $C^{-1}(\mathcal{M})$}\\
    &-\text{\# of generators of $\mathcal{S}_{\mathrm{in}}$ that are in $C^{-1}(\mathcal{M})$}\\
    =&n_{\mathrm{tot}}-k-\text{\# of generators of $C^{-1}(\mathcal{M})$ that anticommute with some elements in $\mathcal{S}_{\mathrm{in}}$}\\
    &-\text{\# of generators of $C^{-1}(\mathcal{M})$ that are in $\mathcal{S}_{\mathrm{in}}$}\\
    =&n_{\mathrm{tot}}-k-n_{\mathcal{M}}.
\end{eqs}
From the second equality to the third, we use the fact that, in terms of generators, we can find a basis such that only its symplectic pair anticommutes with it.  From the third equality to the fourth, we use the fact proved above that the number of generators of $C^{-1}(\mathcal{M})$ that are in $\mathcal{S}_{\mathrm{in}}$ is equal to that of generators of $C^{-1}(\mathcal{M})$ that commute with $\mathcal{S}_{\mathrm{in}}$. Therefore, the two terms exhaust all the generators in $C^{-1}(\mathcal{M})$. On the other hand, the number of detector generators cannot be determined without the detailed structure of the circuit, so we assume it to be $n_{\mathrm{de}}$. Together, we have that the number $n_{\mathcal{S}}$ of generators of the spacetime stabilizer $\mathcal{S}$ is the sum of two types of generators, $n_{\mathcal{S}}=n_{\mathrm{tot}}-k-n_{\mathcal{M}}+n_{\mathrm{de}}$.

Suppose there are $K$ tensors in the circuit, and each tensor has $l_k$ legs. So to determine each tensor, one needs $l_K$ gauge generators. Therefore, the gate gauge has $n_{\mathcal{R}}=\sum_K l_K$ generators. These tensors, together with the input state, induce $N=\sum_k l_k+n_{\mathrm{tot}}$ spacetime qubits. These $K$ tensors are contracted to become a tensor network with $n_{\mathrm{tot}}$ input legs that contract with input states and $n_{\mathrm{tot}}$ output legs. Therefore, there are $\frac{1}{2}(\sum_K l_K+n_{\mathrm{tot}}-n_{\mathrm{tot}})=\frac{1}{2}(\sum_K l_K)$ pairs of legs contracted and the bond gauge group $\mathcal{B}$ has $n_{\mathcal{B}}=\sum_K l_K$ generators. Let $n_{\mathcal{G}}$ be the number of independent generators of the spacetime gauge group $\mathcal{G}=\langle \mathcal{S}_{\mathrm{in}}, \mathcal{B},\mathcal{R},\mathcal{M}\rangle$.  Each detector of the spacetime stabilizer implies a degenerate generator in $\mathcal{G}$, because it can be generated by either $\mathcal{S}_{\mathrm{in}}\cup\mathcal{R}$ or $\mathcal{B}\cup\mathcal{M}$. Therefore, with the number of detectors $n_{\mathrm{de}}$,  $n_{\mathcal{G}}=n_{\mathcal{B}}+n_{\mathcal{R}}+n_{\mathrm{tot}}-k+n_{\mathcal{M}}-n_{\mathrm{de}}$. Altogether, we have 
\begin{eqs}
    \text{\# of gauge qubits}&=\frac{1}{2}(n_{\mathcal{G}}-n_{\mathcal{S}}),\\
    \text{\# of logical qubits}&=\text{\# of total qubits}- \text{\# of gauge qubits} -\text{\# of stabilizer generators}\\&=N-\frac{1}{2}(n_{\mathcal{G}}-n_{\mathcal{S}})-n_{\mathcal{S}}=N-\frac{1}{2}(n_{\mathcal{G}}+n_{\mathcal{S}})\\
    &=\sum_K l_K+n_{\mathrm{tot}}-\frac{1}{2}(n_{\mathcal{B}}+n_{\mathcal{R}}+n_{\mathrm{tot}}-k+n_{\mathcal{M}}-n_{\mathrm{de}}+n_{\mathrm{tot}}-k-n_{\mathcal{M}}+n_{\mathrm{de}})=k.
\end{eqs}

The representatives of the spacetime logical operators can be chosen as the Pauli operators supported on $Q_{\mathrm{input}}$ that commute with $\mathcal{S}_{\mathrm{in}}$. Indeed, they do commute with the stabilizers of the spacetime code, which can be verified explicitly following Lemma~\ref{lemma:spacetimestabilizer}.

\end{proof}

In our framework of the spacetime code, the propagation of spacetime errors is described by gauge equivalence. 
The following lemma establishes the fact that every spacetime error can be propagated to the output boundary.

\begin{definition}[Gauge equivalence]
For any spacetime Pauli operators $E$ and $F$, if there exists a spacetime Pauli operator $G\in \mathcal{G}$, such that $E=FG$, up to a phase, then $E$ is gauge equivalent to $F$, denoted as $E\sim_{\mathcal{G}}F$.
\end{definition}
It is obvious that if $E\sim_{\mathcal{G}}F_1$ and 
$E\sim_{\mathcal{G}}F_2$, then $F_1\sim_{\mathcal{G}}F_2$ and $F_1^\dagger F_2\in \mathcal{G}$. So $F_1^\dagger F_2\sim_ \mathcal{G}\mathbb{I}$. 

\begin{lemma}\label{lemma:outputequiv}
For any Pauli operator $E$ supported on the spacetime qubits $\mathcal{L}$, there exists a Pauli operator $F\sim_{\mathcal{G}}E$ and $\mathrm{supp}(F)\subseteq Q_{\mathrm{output}}$.  $F$ is not unique. Suppose $E\sim_{\mathcal{G}}F_1$ and $E\sim_{\mathcal{G}}F_2$, with $\mathrm{supp}(F_1)\subseteq Q_{\mathrm{output}}$ and $\mathrm{supp}(F_2)\subseteq Q_{\mathrm{output}}$. Then $F_1 F_2^\dagger\in \mathcal{G}$ and  $\mathrm{supp}(F_1 F_2^\dagger)\subseteq Q_{\mathrm{output}}$, so $F_1 F_2^\dagger\in \mathcal{M}$. 
\end{lemma} 
\begin{proof}
Consider a single Pauli operator supported on a specific spacetime location $q_o^{(t)}$ or $q_i^{(t)}$. If this operator is on $q_o^{(t)}$, if $t=t_o(q)$, then it is done. If $t_i(q)\leq t<t_o(q)$, one can always use the bond gauge and the gate gauge to show that $P_o^{(t)}(e_q)\sim_{\mathcal{G}}P_i^{(t+1)}(e_q)\sim_{\mathcal{G}} U_{t+1}P_o^{(t+1)}(e_q) U_{t+1}^\dagger$, The support of $ U_{t+1}P_o^{(t+1)}(e_q) U_{t+1}^\dagger$ is generally on multiple qubits. Because $U_{t+1}$ is Clifford, $(U_{t+1}P_o^{(t+1)}(e_q) U_{t+1}^\dagger)$ is a product of Pauli operators support on time $t+1$. If $\exists q^\prime\in\mathrm{supp} (U_{t+1}P_o^{(t+1)}(e_q) U_{t+1}^\dagger)$, $t_o(q^\prime)=t+1$, then the Pauli operators supported on $q^\prime$ are already on $Q_{\mathrm{output}}$. For the rest part in $\mathrm{supp} (U_{t+1}P_o^{(t+1)}(e_q) U_{t+1}^\dagger)$, one can iterate the above argument to use bond and gate gauges to propagate it to $t+2$, etc. Because each circuit wire has an end, we can finally obtain a gauge-equivalent form of $P_o^{(t)}(e_q)$ on $Q_{\mathrm{output}}$. If this operator is on $q_i^{(t)}$, one can use the gate gauge to show that $P_i^{(t+1)}(e_q)\sim_{\mathcal{G}} U_{t+1}P_o^{(t+1)}(e_q) U_{t+1}^\dagger$ and the rest of the proof is the same as above. Since every single Pauli operator has an output boundary gauge equivalent form, so does their product, which is an arbitrary spacetime Pauli operator $E$.

\end{proof}

Although a spacetime Pauli operator may have different gauge-equivalent forms supported on $Q_{\mathrm{output}}$, we show by the following lemma that they have the same syndrome with respect to the spacetime stabilizers.

\begin{lemma}\label{lemma:spacetimephysical}
For a given circuit described by the gauge group $\mathcal{G}$, suppose its stabilizer group has $n_{\mathcal{S}}$ generators, $\mathcal{S}=\langle S_{1},\ldots, S_{n_{\mathcal{S}}} \rangle$.
Let $S_{k,\mathrm{out}}$ denote the restriction of the spacetime stabilizer $S_{k}$ to the output boundary.
That is, if  $S_{k}$ takes the form of Eq. \eqref{eq:spacetimestabilizer}, then $S_{k,\mathrm{out}}\equiv\bigotimes_{q=1}^{n_{\mathrm{tot}}}\mathcal{U}_{t} P_{o}^{(t_{o}(q))}(e_q)\mathcal{U}_{t}^\dagger$. Define the syndrome vector of a spacetime Pauli operator $E$ as $\sigma(E)\in \mathbb{F}_2^{n_{\mathcal{S}}}$, such that $\sigma(E)_k\equiv[E,S_{k}]$, where $[E,S_{k}]=0$ if $E$ and $S_{k}$ commute, $[E,S_{k}]=1$ if $E$ and $S_{k}$ anti-commute. 
Denote a particular boundary gauge equivalent form as of $E$ as $F$, with $\mathrm{supp}(F)\subseteq Q_{\mathrm{output}}$. Then for any spacetime Pauli operator $E$,
\begin{equation}
    \sigma(E)_k=\sigma(F)_k=[F,S_{k,\mathrm{out}}],
\end{equation}
and $\sigma(E)$ is independent of the choice of $F$.
\end{lemma}
\begin{proof}
Because $E\sim_{\mathcal{G}}F$, write $E=FG$ with $G\in \mathcal{G}$. Then for any $k=1,\ldots,n_{\mathcal{S}}$, $S_k E S_k^{-1} E^{-1}=S_k FG S_k^{-1} G^{-1} F^{-1}=S_k F S_k^{-1}  F^{-1}= S_{k,\mathrm{out}} F S_{k,\mathrm{out}}^{-1}  F^{-1}$. The third equation uses the fact that any spacetime stabilizer commutes with gauge operators. The last equation uses that both $\mathrm{supp}(F)\subseteq Q_{\mathrm{output}}$ and $\mathrm{supp}(S_{k,\mathrm{out}})\subseteq Q_{\mathrm{output}}$. If $E\sim_{\mathcal{G}}F_1$ and $E\sim_{\mathcal{G}}F_2$ with $\mathrm{supp}(F_1)\subseteq Q_{\mathrm{output}}$ and $\mathrm{supp}(F_2)\subseteq Q_{\mathrm{output}}$. Then $F_1\sim_{\mathcal{G}}F_2$. Following the same proof as above, we see that $\sigma(E)_k=\sigma(F_2)_k=\sigma(F_1)_k=[F_1,S_{k,\mathrm{out}}]=[F_2,S_{k,\mathrm{out}}]$. So the syndrome vector is independent of the choice of output boundary gauge equivalent form.

\end{proof}

This lemma gives the syndrome of spacetime stabilizers an explicit physical interpretation. Spacetime stabilizers are not physically measurable operators. What we can measure are the operators supported on $Q_{\mathrm{output}}$. These include the measurement operators generating $\mathcal{M}$, which are measured in the circuit, as well as stabilizer and logical operators of the output code, which may be measured in a subsequent circuit or virtually in the analysis of gadget correctness. On the other hand, the measurement results represent the syndromes of the faults at the time of measurement, which are propagated from faults that occurred earlier in the circuit. The exact spacetime locations at which the faults first occurred cannot be exactly determined. Nevertheless, this lemma shows that the syndromes of a spacetime Pauli operator with respect to spacetime stabilizers are equivalent to the syndrome of the propagated faults at the time of measurement with respect to the measured operator. Therefore, the spacetime code syndromes faithfully represent the circuit measurement outcome. 

In practical circuits, measurements are often performed on each ancilla qubit. In these cases, we can further deduce the form of undetectable spacetime errors. 
\begin{corollary}\label{coro:singlemeasreduc}
For any Pauli operator $E$ supported on the spacetime qubits $\mathcal{L}$, if $E$ is undetectable, that is $\sigma(E)=\mathbf{0}$, and if each measurement gauge generator is a single qubit operator, then there exists a Pauli operator $F\sim_{\mathcal{G}}E$ and $\mathrm{supp}(F)\subseteq Q_{\mathrm{output}} \backslash\Lambda_{\mathrm{meas}}$. 
    
\end{corollary}
\begin{proof}
From Lemma \ref{lemma:outputequiv}, there exists $F^\prime\sim_{\mathcal{G}}E$ and $\mathrm{supp}(F^\prime)\subseteq Q_{\mathrm{output}}$. Now consider the restriction of $F^\prime$ on $\Lambda_{\mathrm{meas}}$. Because $\sigma(F^\prime)=\sigma(E)=0$, we have $[F^\prime|_{\Lambda_{\mathrm{meas}}},m]=0,~\forall m\in \mathcal{M}$. But by assumption, the generators of $\mathcal{M}$ are all single qubit Pauli operators, so $F^\prime|_{\Lambda_{\mathrm{meas}}}\in\mathcal{M}$. Therefore, we can define $F\equiv (F^\prime|_{\Lambda_{\mathrm{meas}}})^\dagger F^\prime$, then $F\sim_{\mathcal{G}}E$ and $\mathrm{supp}(F)\subseteq Q_{\mathrm{output}} \backslash\Lambda_{\mathrm{meas}}$.
\end{proof}

This corollary can be viewed as a certain form of the cleaning lemma, in the sense that we can clean each undetectable error configuration to have its gauge equivalent form only supported on the unmeasured output boundary.

\subsection{Spacetime complex}
\begin{definition}[Spacetime complex]
We call the chain complex of a circuit $\mathcal{A}$
\begin{equation}
\label{eq:spacetime_complex}
\begin{tikzcd}[row sep=1.2em, column sep=1.0em]
A_2 \arrow[r, "\partial_{2,A}"] &
A_1 \arrow[r, "\partial_{1,A}"] &
A_0, \\
{\scriptstyle \text{gauge checks}} &
{\scriptstyle \text{fault positions}} &
{\scriptstyle \text{stabilizer syndromes}}
\end{tikzcd}
\end{equation}
as the spacetime complex of a circuit, whose boundary maps are defined by the matrix representation of generators of gauge group $\mathcal{G}$ and stabilizer group $\mathcal{S}$ respectively
\begin{eqs}
    \partial_{2,A}= G^\mathsf{T},\quad \partial_{1,A}=S \Omega,
\end{eqs}
and $\Omega=\begin{pmatrix}
    0 & \mathbb{I} \\
    -\mathbb{I} & 0
\end{pmatrix}$ is the symplectic form encoding the commutation relations between Pauli $X$ and $Z$ operators. The rows of $G$ and $S$ are the symplectic representations of the corresponding Pauli operators. The spacetime stabilizers are generated by row vectors in $\partial_{1,A}$. 
\end{definition}

\begin{definition}[Fault distance of spacetime complex]
    The fault distance of a spacetime complex $\mathcal{A}$ is defined as its 1-systolic distance $d_\st \equiv\min\{|\mathbf{x}| : \mathbf{x} \in \mathrm{ker}\partial_{1,A} \backslash \mathrm{im} \partial_{2,A}\}$, which are undetectable faults with respect to the spacetime subsystem code stabilizers.
\end{definition}

As in previous treatments of fault complexes, not every spacetime stabilizer constraint represented by $\partial_{1,A}$ is actually measured in the circuit. We therefore restrict attention to the syndromes of spacetime stabilizers that are detectors, following Definition \ref{def:detector_stabilizer_tube}. It is because detectors can be generated from either $\mathcal{S}_{\mathrm{in}}\cup\mathcal{R}$ or $\mathcal{B}\cup\mathcal{M}$, that we can deduce a fixed measurement outcome after fixing the gauges in $\mathcal{S}_{\mathrm{in}}$, $\mathcal{R}$, and $\mathcal{B}$ when no errors occur in the circuit.

\begin{definition}[Detector group]
The detector group $\mathcal{D}$ is a subgroup of the spacetime stabilizer group $\mathcal{S}$,
    \begin{eqs}
        \mathcal{D}\equiv\{s\in \mathcal{S}|~s|_{\mathrm{input}}\in \mathcal{S}_{\mathrm{in}},~ s|_{\mathrm{output}}\in \mathcal{M}\}.
    \end{eqs}
The detector matrix $D$ is defined as the binary symplectic representation of a set of generators of $\mathcal{D}$. 
\end{definition}
    
\begin{definition}[bulk fault positions]
The set of bulk fault positions $J$ is the set of spacetime qubits on which at least one detector has nonzero support $J\equiv\operatorname{colsupp}(D \Omega)$. Correspondingly, the linear subspace $A_{1,\mathrm{bulk}}$ is defined as $A_{1,\mathrm{bulk}}\equiv\mathbb{F}_2^{J}$. Because the detectors also have support on the ancilla initialization and measurement positions
\(\Lambda_{\mathrm{meas}}\), the set \(J\) includes both the qubits in the circuit bulk and the ancilla initialization and measurement positions.
\end{definition}
The set of bulk fault positions is also referred to as the \emph{detecting region} in the circuit detector formalism \cite{mcewen2023relaxing}. 

\begin{definition}[Input, bulk, and output regions]
    The fault position space $A_1$ can be decomposed into three parts
    \begin{eqs}
         A_1 = A_{1,\mathrm{input}} \oplus A_{1,\mathrm{bulk}} \oplus A_{1,\mathrm{output}},
    \end{eqs}
    where the fault positions in $A_{1,\mathrm{input}},A_{1,\mathrm{bulk}},A_{1,\mathrm{output}}$ form the input, bulk, and output regions respectively.  Note that $Q_{\mathrm{input}}$ is a subset of input region, and $Q_{\mathrm{output}} \backslash \Lambda_{\mathrm{meas}}$ is a subset of output region.
\end{definition}  
For example, consider a CSS code syndrome extraction circuit in which the $X$- and $Z$-syndrome extraction rounds alternate in time. For the $X$-syndrome extraction, we define the bulk $Z$-fault positions as the collection of all the $Z$-fault positions between the first $X$-syndrome extraction round and the last $X$-syndrome extraction round, respectively. Similarly, for $Z$-syndrome extraction, we define the bulk $X$-fault positions as a collection of all the $X$-fault positions between the first and the last $Z$-syndrome extraction rounds. The motivation for introducing bulk gauge checks is that we wish to characterize the intrinsic fault-tolerance properties of the channel independently of the particular input and output boundaries. 

\begin{definition}[bulk gauge group]    
The bulk gauge group $\mathcal{G}_{\mathrm{bulk}}$ is defined as the gauge operators supported entirely within the bulk region
    \begin{eqs}
        \mathcal{G}_{\mathrm{bulk}}:=\left\{g\in\mathcal{G}:\operatorname{supp}(g)\subseteq J\right\}.
    \end{eqs}
The bulk gauge check matrix  $G_{\mathrm{bulk}}$ is defined as the binary symplectic representation of a set of generators of $\mathcal{G}_{\mathrm{bulk}}$. 
\end{definition}
By construction, both 
$D$ and $G_{\mathrm{bulk}}$ have zero columns outside $J$. With a slight abuse of notation, we denote detector and bulk gauge check matrices restricted to $A_{1,\mathrm{bulk}}$ by $D$ and $G_{\mathrm{bulk}}$ after removing all the zero columns outside $J$.

\begin{definition}[Bulk complex]
    The bulk complex is a subcomplex of a spacetime complex $\mathcal{A}$  
    \begin{equation}
\label{eq:bulk_subcomplex}
\begin{tikzcd}[row sep=1.2em, column sep=1.0em]
A_{2,\mathrm{bulk}} \arrow[r, "G_{\mathrm{bulk}}^\mathsf{T}"] &
A_{1,\mathrm{bulk}} \arrow[r, "D \Omega"] &
A_{0,\mathrm{bulk}}. \\
{\scriptstyle \text{bulk gauge checks}} &
{\scriptstyle \text{bulk fault positions}} &
{\scriptstyle \text{bulk detector syndromes}}
\end{tikzcd}
\end{equation}

\end{definition}

The bulk gauge checks and detectors satisfy  $G_{\mathrm{bulk}}^\mathsf{T} A_{2,\mathrm{bulk}} \subseteq A_{1,\mathrm{bulk}}$ and $\Omega D^\mathsf{T} A_{0,\mathrm{bulk}} \subseteq A_{1,\mathrm{bulk}}$ which are gauge and stabilizer elements that only acts on the bulk qubits $A_{1,\mathrm{bulk}}$\footnote{This can also be interpreted as puncturing or shortening a quantum code.}. Moreover, because the detector group, as a subgroup of $\mathcal{S}$, is in the center of the spacetime gauge group $\mathcal{G}$, hence $D \Omega G_{\mathrm{bulk}}^{\mathsf{T}}=0$. So Eq.~\eqref{eq:bulk_subcomplex} is indeed a chain complex.

\begin{definition}[Minimum-weight recovery]
   A minimum-weight recovery is a map that takes a syndrome $\sigma\in\operatorname{im}\partial$ and returns a recovery operation $\widetilde{\varepsilon}$ satisfying $|\widetilde{\varepsilon}|\leq|\varepsilon|$ for every fault configuration $\varepsilon$ with $\partial\varepsilon=\sigma$. The fault and its recovery satisfy $\partial(\varepsilon+\widetilde{\varepsilon})=0$, where $\partial$ is a check matrix, which can be a detector matrix or a spacetime stabilizer matrix. Here $|\cdot|$ denotes the Hamming weight of a binary vector.
\end{definition}

\begin{definition}[Bulk effective distance]
    For a gadget with minimum-weight recovery, the bulk effective distance is defined as
    \begin{eqs}
        d_{\mathrm{bulk}} := \min\{|\mathbf{x}|: (\mathbf{x}+\widetilde{\mathbf{x}})\in \mathrm{ker} D\Omega, (\mathbf{x}+\widetilde{\mathbf{x}}) \in \mathrm{ker}\partial_{1,A}{\backslash} \mathrm{im} \partial_{2,A}, \mathrm{supp}(\mathbf{x})\subseteq J\},
    \end{eqs}
    i.e., the minimal weight of a bulk undetectable fault that induces a nontrivial logical action after a minimum-weight recovery $\widetilde{\mathbf{x}}$ acting on the bulk, satisfying $\mathrm{supp}(\widetilde{\mathbf{x}}) \subseteq J$ and $D \Omega \widetilde{\mathbf{x}}= D \Omega \mathbf{x}$. 
\end{definition}

\begin{remark}
    The bulk effective distance defined above depends on the choice of decoder and is stated here for a minimum-weight decoder. A different recovery rule may yield a different bulk effective distance. To adapt this definition to a specific decoder, one can replace the minimum-weight recovery $\widetilde{\mathbf{x}}$ with the recovery operation $\mathbf{x}'=R_{D\mathbf{x}}$ returned by the chosen decoder $R$.
\end{remark}

\begin{lemma}
    If there exists a minimum-weight undetectable fault $F\in \mathrm{ker} \partial_{1,A} {\backslash} \mathrm{im} \partial_{2,A}$ such that $\mathrm{supp}(F) \subseteq J$, then the bulk effective distance is upper-bounded by the spacetime fault distance: $2d_{\mathrm{bulk}}+1\leq d$.
\end{lemma}
For many standard fault-tolerant schemes, including repeated syndrome extraction and single-shot error correction, a minimum-weight logical representative can be chosen to have support entirely within the bulk. Consequently, for these schemes, the spacetime fault distance provides an upper bound on the bulk effective distance.

Our definition of bulk effective distance resembles that of the fault distance in Ref.~\cite{beverland2024fault}. A key difference is that the bulk effective distance here includes the effect of minimum-weight recovery depending on detector syndromes. Note that not every undetectable fault induces a logical action. For example, faults that occur after the last round of syndrome extraction are undetectable in the present circuit, but they do not induce logical actions as long as their weight is less than the distance of the output code. These residual errors are detectable in the next syndrome extraction circuit.

\begin{lemma}[Bulk detectable faults]\label{lemma:bulk_detectable_distance}
    For a spacetime complex with bulk effective distance
    $d_{\mathrm{bulk}}=t$, every bulk fault $F$ satisfying
    $|F|\leq 2t$ either produces a nontrivial detector
    syndrome or has a trivial logical action. 
\end{lemma}
\begin{proof}
    Let $F$ be a bulk fault with $|F|\leq 2t$. If $F$
    has a nontrivial detector syndrome, then it is detectable. It therefore remains to consider the case in which $F$ has trivial detector syndrome.

    Partition the elementary faults in $F$ into two fault configurations
    $F_1$ and $F_2$ such that
    \[
    F=F_1F_2,
    \qquad
    |F_1|,
    |F_2|\leq t.
    \]
    Since the detector syndrome is linear and $F$ has trivial syndrome,
    $F_1$ and $F_2$ produce the same detector syndrome $D \Omega F_1=D \Omega F_2$. Let $\widetilde{F_1},\widetilde{F_2}$ be the minimum-weight recovery for detector syndrome $D \Omega F_1$ and $D \Omega F_2$, given the syndromes are equal $D \Omega F_1=D \Omega F_2$, the recovery operations are also the same $\widetilde{F_1}=\widetilde{F}_2$. By the definition of
    $d_{\mathrm{bulk}}=t$, every bulk fault of weight at most $t$ is
    correctable, thus 
    \begin{eqs}
        \widetilde{F_1} F_1 \simeq \mathbb{I},\quad \widetilde{F_2} F_2 \simeq \mathbb{I},
    \end{eqs}
    where $\simeq$ means that the two operators have the same logical action.
    Taking the adjoint of the first expression and multiplying it by the second gives
    \begin{eqs}
        (\widetilde{F_1} F_1)^\dagger \widetilde{F_2} F_2= F_1^\dagger \widetilde{F_1}^\dagger \widetilde{F_2} F_2= F_1^\dagger F_2 \simeq \mathbb{I}.
    \end{eqs}
    Since qubit Pauli operator is self-adjoint, then
    \begin{eqs}
        F_1^\dagger F_2 = F_1 F_2= F\simeq \mathbb{I}
    \end{eqs}
    has trivial logical action. Therefore, every bulk fault of weight at
    most $2t$ is either detectable or logically trivial.
\end{proof}

\section{Bridging fault distance to Gadget Correctness properties}

In this section, we present three theorems systematically connecting the fault distances of spacetime complexes to Gottesman's gadget correctness properties~\cite[Sec.~10.2]{gottesman2024surviving}.

\begin{definition}[correctness property of a circuit]
A circuit satisfies the correctness property if the input code has code distance $d=2t+1$, and if the number of errors in the input state together with the number of faults in the circuit is at most $t$, then after perfect error correction of the output code, the outcome is equivalent to an ideal circuit acting on the errorless codewords. 
\end{definition}

\begin{theorem}\label{thm:ctod}Let a Clifford circuit have the input state as a codeword of an $[\![n,k,d]\!]$ stabilizer code, whose code distance is $d=2t+1$, together with ancilla qubits in initial states. If this circuit satisfies Gottesman's correctness property for $t$ faults, then the induced spacetime subsystem code has fault distance $d_\st=d$.   
\end{theorem}

\begin{proof}
A circuit may have multiple rounds of adaptive Pauli feedback, which applies a Pauli operator depending on previous measurement results. Suppose that there are $n_{\mathrm{fb}}$ rounds of feedback. At time \(t_{j}\), the circuit applies a Pauli feedback \(R_j\) determined by the preceding measurement outcomes. The stabilizer circuit property guarantees that each such correction can be propagated through the remaining gates and replaced by an equivalent Pauli correction at the output. Propagating a correction through a later Pauli measurement may flip the interpretation of that measurement outcome, but this can be handled by updating the classical processing rule.
Therefore, a circuit containing several rounds of adaptive Pauli feedback is logically equivalent to a circuit in which all feedback is postponed until the output. The postponed feedback on the output can be written as
\begin{eqs}
R_{\mathrm{tot}}(\mathbf m)
=
\prod_{j=1}^{n_{\mathrm{fb}}}
\Phi_{t_{j}\rightarrow\tau}\!\left(R_j(\mathbf m_{<j})\right),
\end{eqs}
where \(\Phi_{t_{j}\rightarrow\tau}\) propagates a Pauli operator from time \(t_{f_j}\) to the output time \(\tau\), and $\mathbf{m}_{<j}$ denotes all the measurement outcome up to time $t_j$. The perfect output error correction appearing in Gottesman’s correctness property can similarly be represented by perfect measurements of the output code stabilizers followed by a Pauli correction. Let $\sigma_{\mathrm{out}}$ denote the perfect syndrome measurement outcome of the output code. Combining the minimum-weight recovery of output syndrome $\sigma_{\mathrm{out}}$ with $R_{\mathrm{tot}}(\mathbf{m})$ gives an overall feedback, denoted as $\rho(\mathbf{m},\sigma_{\mathrm{out}})$. 

To compare the spacetime code framework with that of the correctness property, we first note that an error on the input state can be regarded as a spacetime fault occurring immediately before the first circuit operation. Thus, \(r\) input errors together with \(s\) circuit faults correspond to a spacetime fault \(\mathbf e\) of weight $|\mathbf{e}|=r+s$. By Lemma~\ref{lemma:spacetimephysical}, the spacetime stabilizer syndromes \(\partial_{1,A}\mathbf e\) are equivalent to the syndrome obtained from the measurements of the logical operators, circuit detectors, and the output code stabilizer (See Lemma\ref{lemma:spacetimestabilizer} and Definition \ref{def:detector_stabilizer_tube}). Hence, all the measurement outcomes are equivalent to the spacetime stabilizer syndromes, which are both functions of the spacetime fault configurations. Therefore, we write the feedback in the correctness verification as $\rho(\mathbf{m}(\mathbf{e}),\sigma_{\mathrm{out}}(\mathbf{e})) \equiv \rho (\partial_{1,A} \mathbf{e})$. Gottesman’s correctness property implies that for every spacetime fault \(\mathbf e\) with \(|\mathbf e|\le t\), the residual fault $\mathbf e+\rho\!\left(\partial_{1,A}\mathbf e\right) 
$ has trivial logical action compared to the ideal circuit Pauli feedback $\rho(\mathbf{m}(0),\sigma_{\mathrm{out}}(0))$ . 

We now prove that \(d_\st\ge 2t+1\). Suppose, to the contrary, that there exists a nontrivial undetectable spacetime fault \(\boldsymbol{\ell}\) with $|\boldsymbol{\ell}| \leq 2t$. Partition its support into two disjoint fault configurations \(\mathbf e_1\) and \(\mathbf e_2\) satisfying
\begin{eqs}
\boldsymbol{\ell}=\mathbf e_1+\mathbf e_2,
\qquad
|\mathbf e_1|,|\mathbf e_2|\le t.
\end{eqs}
Because \(\boldsymbol{\ell}\) is undetectable, we have $\partial_{1,A}\boldsymbol{\ell}=0$.
Consequently, $\mathbf{e}_1, \mathbf{e}_2$ have the same spacetime stabilizer syndrome
\begin{eqs}
    \partial_{1,A}\mathbf e_1
=
\partial_{1,A}\mathbf e_2,
\end{eqs}
and will be assigned the same recovery \(\rho\). 
Since both faults have weight at most \(t\), the correctness property implies that both residual faults $
\mathbf e_1+\rho(\partial_{1,A} \mathbf{e_1})$ and $
\mathbf e_2+\rho(\partial_{1,A} \mathbf{e}_2)$ have trivial logical action compared to the ideal feedback $\rho(\mathbf{m}(0),\sigma_{\mathrm{out}}(0))$, then so does their sum. But 
\begin{eqs}
\mathbf e_1+\rho(\partial_{1,A} \mathbf{e_1})+
\mathbf e_2+\rho(\partial_{1,A} \mathbf{e}_2) = \mathbf{e}_1+ \mathbf{e}_2= \boldsymbol{\ell},
\end{eqs}
this contradicts the assumption that \(\boldsymbol{\ell}\) is a nontrivial spacetime logical fault. Therefore, we have the lower bound of $d_\st$
\begin{eqs}
    d_\st\ge 2t+1=d.
\end{eqs}
For the upper bound, put a minimum-weight nontrivial logical Pauli operator \(L\) of the input code with $|L|=d$ on $Q_{\mathrm{input}}$. By Lemma~\ref{lemma:spacetimelogical}, this produces a nontrivial undetectable spacetime fault of the same weight. Consequently, we can obtain the upper bound
$d_\st \le d$. 
Combining the two bounds gives $d_\st =d$.

\end{proof}

\begin{theorem}\label{thm:GCP}
   Let a Clifford circuit with no Pauli feedback have the input state as a codeword of an $[\![n,k,d]\!]$ stabilizer code, whose code distance is $d=2t+1$, together with ancilla qubits in initial states. If the ideal circuit implements the designated logical operation, then the induced spacetime subsystem code has fault distance $d_\st =2t+1$ if and only if this circuit satisfies the correctness property. 
\end{theorem}

\begin{proof}
We prove the correctness property implies spacetime code distance in Theorem \ref{thm:ctod}, taking the no-feedback case. To prove the other direction, suppose otherwise; then there exists a fault $\mathbf{f}$ of weight $|\mathbf{f}|\leq t$ such that after the perfect error correction, the composite $\mathbf{f}+\rho(\mathbf{m}(\mathbf{f}),\sigma_{\mathrm{out}}(\mathbf{f}))$ has a non-trivial logical action. However, this contradicts the assumption that the spacetime code has distance $d_\st =2t+1$, so each fault with weight smaller than $t$ is correctable. Therefore, the spacetime code distance also implies the correctness of the circuit.
\end{proof}

\begin{definition}[Fault non-amplifying syndrome extraction]
A syndrome extraction gadget is called \(t\)-fault-non-amplifying if, for all $\varepsilon\in A_{1,\mathrm{input}}$ and
$\eta\in A_{1,\mathrm{bulk}}$ satisfying
$|\varepsilon|+|\eta|\leq t$, there exist a pure measurement fault $ m(\varepsilon,\eta )\in A_{1,\mathrm{bulk}}$ and output fault $\varepsilon_{\mathrm{out}}(\varepsilon,\eta)$ supported on $Q_{\mathrm{output}}\backslash \Lambda_{\mathrm{meas}}$ such that
\begin{eqs}
    \varepsilon \oplus \eta \oplus \mathbf{0} \sim_{\mathcal{G}}  \mathbf{0} \oplus m(\varepsilon,\eta) \oplus \varepsilon_{\mathrm{out}}(\varepsilon,\eta),
\end{eqs}
where the equivalence is taken modulo the full spacetime gauge group. Moreover, the weight of the output boundary component does not exceed the total weight of the original fault configuration:
\begin{eqs}
    |\varepsilon_{\mathrm{out}} (\varepsilon,\eta)| \leq |\varepsilon| + |\eta|.
\end{eqs}
\end{definition}

\begin{theorem}\label{thm:ECCP}
    The spacetime complex for a $t$-fault-non-amplifying syndrome extraction circuit with minimum-weight recovery based on detector syndromes satisfies Gottesman's FT Error Correction Correctness Property for every fault configuration of total weight $r\leq t$, if the spacetime complex has bulk effective distance $d_{\mathrm{bulk}} \geq t$ and spacetime fault distance $d_\st  \geq 2t+1$.
\end{theorem}

\begin{proof}
    Suppose we have a weight-$r$ fault, which can be written as $\varepsilon_1 \oplus \varepsilon _2 \oplus \varepsilon_3$ with $|\varepsilon_1|+|\varepsilon_2|+|\varepsilon_3|=r$ where $r\leq t$. Here the three components are supported on the input, bulk, and output region. Let such spacetime complex has fault distance $d \geq 2t+1$ and bulk effective distance $d_{\mathrm{bulk}} \geq t$. 
    
    By applying the minimum-weight recovery to the bulk syndrome of $\varepsilon_2$, we apply a recovery operation $\widetilde{\varepsilon_2}$ with $|\widetilde{\varepsilon_2}| \leq |\varepsilon_2|$. After we apply the bulk recovery $\widetilde{\varepsilon_2}$, the fault is written as $\varepsilon_1 \oplus (\varepsilon_2 + \widetilde{\varepsilon_2})\oplus \varepsilon_3$.

    We now analyze the boundary action of the actual bulk fault and the applied recovery separately. Since $|\varepsilon_1| + |\varepsilon_2| \leq t$, the \(t\)-fault-non-amplifying syndrome extraction gadget gives
    \begin{eqs}\label{eq:ECCP_fault_equivalence1}
        \varepsilon_1 \oplus \varepsilon_2 \oplus \mathbf{0} \sim_{\mathcal{G}} \mathbf{0} \oplus m(\varepsilon_1, \varepsilon_2) \oplus \varepsilon_{\mathrm{out}} (\varepsilon_1, \varepsilon_2),
    \end{eqs}
    where $m(\varepsilon_1, \varepsilon_2)$ is a pure measurement fault and $\varepsilon_{\mathrm{out}} (\varepsilon_1, \varepsilon_2)$ is supported on the output boundary, with $|\varepsilon_{\mathrm{out}} (\varepsilon_1, \varepsilon_2)| \leq |\varepsilon_1| + |\varepsilon_2|$.

    Similarly, because $|\widetilde{\varepsilon_2} | \leq |\varepsilon_2| \leq t$, applying the fault non-amplifying property gives us 
    \begin{eqs}\label{eq:ECCP_fault_equivalence2}
        \mathbf{0} \oplus \widetilde{\varepsilon_2} \oplus \mathbf{0} \sim_{\mathcal{G}} \mathbf{0} \oplus m(0, \widetilde{\varepsilon_2}) \oplus \varepsilon_{\mathrm{out}} (0,\widetilde{\varepsilon_2}),
    \end{eqs}
    where $m(0, \widetilde{\varepsilon_2})$ is a pure measurement fault and $\varepsilon_{\mathrm{out}} (0,\widetilde{\varepsilon_2})$ is supported on the output boundary with $|\varepsilon_{\mathrm{out}} (0,\widetilde{\varepsilon_2})| \leq |\widetilde{\varepsilon_2}| \leq |\varepsilon_2|$.

    Combining Eqs. \eqref{eq:ECCP_fault_equivalence1} and \eqref{eq:ECCP_fault_equivalence2}, the fault configuration after the bulk recovery $\widetilde{\varepsilon_2}$ satisfies
    \begin{eqs}
        \varepsilon_1 \oplus (\varepsilon_2 +\widetilde{\varepsilon_2}) \oplus \varepsilon_3 \sim_{\mathcal{G}} \mathbf{0} \oplus (m(\varepsilon_1, \varepsilon_2) + m(0,\widetilde{\varepsilon_2})) \oplus (\varepsilon_{\mathrm{out}} (\varepsilon_1, \varepsilon_2) + \varepsilon_{\mathrm{out}} (0, \widetilde{\varepsilon_2}) + \varepsilon_3). 
    \end{eqs}
    The contribution $\varepsilon_{\mathrm{out}} (0, \widetilde{\varepsilon_2})$ is completely determined by the known recovery $\widetilde{\varepsilon_2}$. It can therefore be canceled by a deterministic output correction or, equivalently, incorporated into the output Pauli frame. After this Pauli-frame update, the effective residual fault is 
    \begin{eqs}
        \mathbf{0} \oplus (m(\varepsilon_1, \varepsilon_2) + m(0,\widetilde{\varepsilon_2})) \oplus (\varepsilon_{\mathrm{out}}(\varepsilon_1, \varepsilon_2) +\varepsilon_3 ) \equiv \mathbf{0} \oplus (m(\varepsilon_1, \varepsilon_2) + m(0,\widetilde{\varepsilon_2})) \oplus \varepsilon_R,
    \end{eqs}
    where $\varepsilon_R:= \varepsilon_{\mathrm{out}}(\varepsilon_1, \varepsilon_2) +\varepsilon_3$ is the output residual error. The weight of residual error satisfies 
    \begin{eqs}
        |\varepsilon_R| \leq |\varepsilon_{\mathrm{out}} (\varepsilon_1, \varepsilon_2)| + |\varepsilon_3| \leq |\varepsilon_1| + |\varepsilon_2| + |\varepsilon_3| \leq t.
    \end{eqs}

     We then use the matrix form of stabilizer tubes   $S_{\mathrm{in}} \bigoplus \big(\bigoplus_{t=0}^{\tau} S_{t} \big)\bigoplus S_{\mathrm{out}}$ in Definition~\ref{def:detector_stabilizer_tube}, where $S_{\mathrm{in}}, S_{\mathrm{out}}$ are check matrices for input and output code supporting on the input and output boundaries; $S_{t}$ is the check matrices for the instantaneous propagation of $S_{\mathrm{in}}$ at time $t$, to detect the error on the output boundary\footnote{Note that $\oplus_{t=0}^\mathsf{\tau} S_t$ is a direct sum over all the instantaneous propagation of $S_{\mathrm{in}}$ for all the time $t$ between the input and output boundaries.}. Given we do not have any error on the input boundary and $m(\varepsilon_1, \varepsilon_2) + m(0,\widetilde{\varepsilon_2})$ is a pure measurement fault which is not on the support of stabilizer tube, we have the syndrome of stabilizer tubes
     \begin{eqs}
         \left (S_{\mathrm{in}} \bigoplus \big(\bigoplus_{t=0}^{\tau}S_{t} \big)\bigoplus S_{\mathrm{out}} \right) \Omega \cdot  (\mathbf{0} \oplus (m(\varepsilon_1, \varepsilon_2) + m(0,\widetilde{\varepsilon_2}))  \oplus \varepsilon_R)=  S_{\mathrm{out}} \Omega \varepsilon_R.
     \end{eqs}
    Note that the measurement fault $\mathbf{0} \oplus (m(\varepsilon_1, \varepsilon_2) + m(0,\widetilde{\varepsilon_2}))\oplus \mathbf{0}$ has zero contribution, because $\bigoplus_{t=0}^{\tau}S_{t}$ doesn't overlap with measurement ports.

    We then perform a minimum-weight recovery restricted to the output boundary based on the syndrome $ S_{\mathrm{out}} \Omega \varepsilon_R$, which is 
     \begin{eqs}
         \widetilde{\varepsilon_R}:=\argmin \{|\rho|: S_\mathrm{out} \Omega \rho= S_{\mathrm{out}} \Omega \varepsilon_R\},  
     \end{eqs}
     where $S_\mathrm{out}$ is the check matrix of output code. Note that this is a minimum-weight recovery restricted to the output boundary, which corresponds to the ideal decoder acting on the output code in Gottesman's formulation, instead of the minimum-weight recovery on the entire spacetime block.

     After the second recovery, we finish the recovery. Now the fault configuration becomes
     \begin{eqs}
         \mathbf{0} \oplus (m(\varepsilon_1, \varepsilon_2) + m(0,\widetilde{\varepsilon_2})) \oplus (\varepsilon_R +\widetilde{\varepsilon_R}) 
         \cong \mathbf{0} \oplus \mathbf{0} \oplus (\varepsilon_R +\widetilde{\varepsilon_R}) ,
     \end{eqs}
     because we are not going to use the measurement outcome for further recovery and $m(\varepsilon_1, \varepsilon_2) + m(0,\widetilde{\varepsilon_2})$ doesn't act on the codeblock, so we are free to discard the measurement fault part.

     Since the spacetime code has fault distance $2t+1$, the code distance of the output code is at least $2t+1$  and $| \varepsilon_R + \widetilde{\varepsilon_R}|\leq 2t$. Finally, we have $\mathbf{0} \oplus \mathbf{0} \oplus (\varepsilon_R +\widetilde{\varepsilon_R}) \sim_{\mathcal{G}} \mathbf{0 \oplus \mathbf{0 \oplus \mathbf{0}}}$, which matches Gottesman's \emph{Error Correction Correctness Property} (ECCP).
\end{proof}

Theorem~\ref{thm:ECCP} relates the fault distances of a spacetime complex with minimum-weight recovery to the \textit{FT Error Correction Correctness Property} (ECCP) introduced in Ref.~\cite[Sec.~10.2]{gottesman2024surviving}. This result shows that the fault distance effectively quantifies a gadget's fault-tolerance properties. This result indicates that requiring a spacetime code associated with the fault-non-amplifying syndrome extraction circuit to satisfy $ d \geq 2t+1$ and $d_{\mathrm{bulk}}\geq t$
is a stronger condition than ECCP itself. Indeed, ECCP only requires the existence of a suitable recovery procedure, whereas the two-step recovery constructed here need not be optimal among all possible recovery procedures for the spacetime code.

A natural generalization is to replace the non-amplifying bound by a non-decreasing function $f(x)$ and restrict the input-plus-bulk fault weight so that the propagated output error remains within weight $t$.

\begin{definition}[$(t,f)$-bounded fault amplifying syndrome extraction]
Let $f(x): \mathbb{N} \rightarrow \mathbb{N}$ be a non-decreasing function, that characterizes the fault amplification from input boundary to output boundary, satisfying $f(x)\geq x$. A syndrome extraction gadget is called $(t,f)$-bounded fault amplifying if, for every pair of fault configurations $\varepsilon \in A_{1,\mathrm{input}} , \eta \in A_{1,\mathrm{bulk}}$ satisfying $|\varepsilon|+|\eta| \leq t$, there exist a pure measurement fault $m(\varepsilon,\eta) \in A_{1,\mathrm{bulk}} $ and an output fault $\varepsilon_{\mathrm{out}}(\varepsilon,\eta)$ supported on $Q_{\mathrm{output}} \backslash \Lambda_{\mathrm{meas}}$ such that 
\begin{eqs}
    \varepsilon \oplus \eta \oplus \mathbf{0} \sim_{\mathcal{G}} \mathbf{0} \oplus m(\varepsilon,\eta ) \oplus \varepsilon_{\mathrm{out}} (\varepsilon,\eta),  
\end{eqs}
where the equivalence is taken modulo the full spacetime gauge group, and the output component satisfies $|\varepsilon_{\mathrm{out}}(\varepsilon,\eta)| \leq f (|\varepsilon|+ |\eta|)$. For \(f(x)=x\), this definition reduces to the definition of a \(t\)-fault-non-amplifying syndrome extraction gadget.
\end{definition}


\begin{corollary}
\label{cor:amplified_ECCP}
Consider a $(t,f)$-bounded fault amplifying syndrome extraction circuit with minimum-weight recovery based on detector syndromes. Also assume that $x\mapsto f(x)-x$ is nondecreasing. Suppose that its spacetime fault distance satisfies $d\geq 2t+1$ and its bulk effective distance satisfies $d_{\mathrm{bulk}} \geq t$. Then  Gottesman's FT Error Correction Correctness Property holds for every fault configuration of total weight $r$ such that $f(r)\leq t$. 
\end{corollary}
Corollary~\ref{cor:amplified_ECCP} follows by adapting the proof of Theorem~\ref{thm:ECCP}.

\section{Syndrome extraction under phenomenological error model}\label{sec:foliation}

In this section, we derive the detectors of a syndrome extraction circuit of an $N$-qubit CSS code with check matrices $H_X \in \mathbb{F}_2^{r_X \times N}, H_Z \in \mathbb{F}_2^{r_Z \times N}$ from our formalism, by only assuming the phenomenological error model, which provides a concrete example to show how do the detectors of foliation naturally arise from the spacetime complex and circuit. 

When we consider phenomenological error model, we assume the syndrome extraction process is perfect except the measurement, so we consider following tensor that represents the syndrome extraction gadget as depicted in Fig.~\ref{fig:SE_gadget} (a), where we assume the faults can only happen on the edges, but the interior of the gadgets are perfect. The $X$ syndrome extraction tensors are placed at odd time steps, whereas the $Z$ syndrome extraction tensors are placed at even time steps.

\begin{figure}[h]
    \centering
    \includegraphics[width=0.95\textwidth]{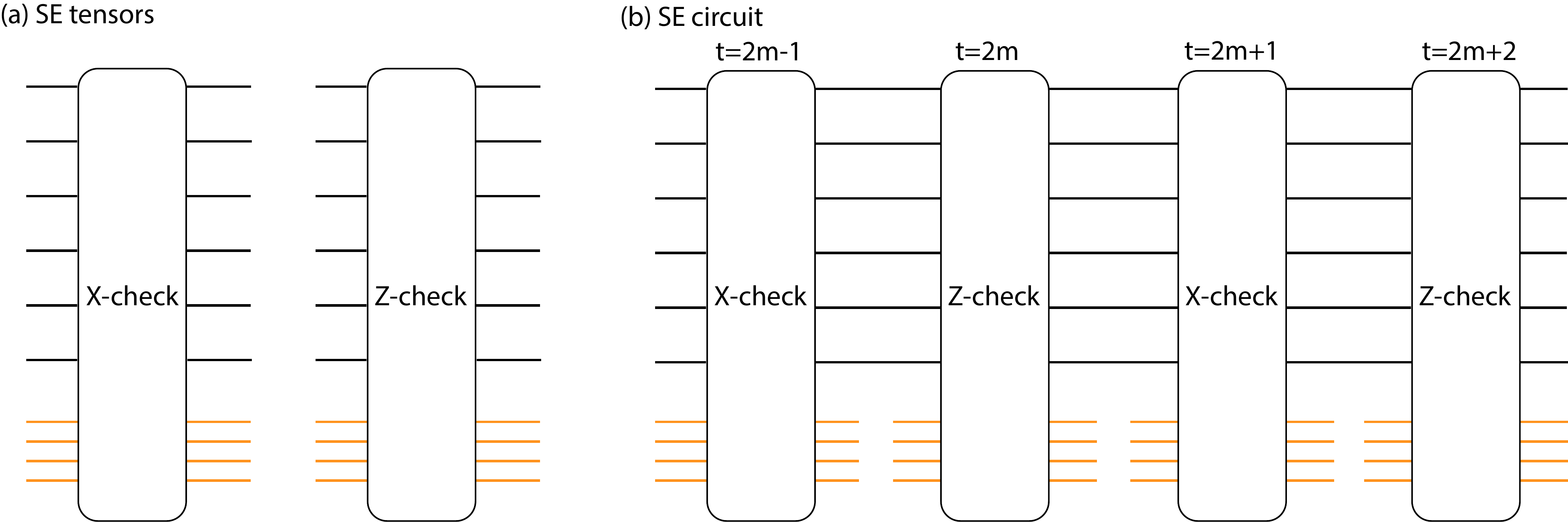}
    \caption{(a)The $X$- and $Z$-syndrome extraction (SE) gadgets. The legs on the left side of gadgets are the input of gadgets, the legs on the right side of gadgets are the output of gadgets, the legs connecting to the bottom of gadgets are measurement qubits for syndrome readout. The legs in black represents the qubits in spacetime of the code block, and legs in orange represents the ancilla qubits for $X$ and $Z$ check measurements; (b) the syndrome extraction (SE) circuit, where we perform $X$- and $Z$-syndrome extraction alternately.}
    \label{fig:SE_gadget}
\end{figure}

For convenience, we write the matrices $H_X, H_Z$ as collections of row vectors, where each individual row vectors corresponds to a check
\begin{eqs}
H_X =
\begin{pmatrix}
\text{---} & h_1       & \text{---} \\
\text{---} & h_2       & \text{---} \\
            & \vdots    &             \\
\text{---} & h_{r_X}   & \text{---}
\end{pmatrix},\quad 
H_Z =
\begin{pmatrix}
\text{---} & g_1       & \text{---} \\
\text{---} & g_2       & \text{---} \\
            & \vdots    &             \\
\text{---} & g_{r_Z}   & \text{---}
\end{pmatrix}.
\end{eqs}

At each time slice, we label the physical qubits on the code block by numbers from 1 to $N$, and the ancilla qubits that will be used for $X$- and $Z$- syndrome extraction are labeled by numbers from $N+1$ to $N+r_X$ and $N+1$ to $N+r_Z$ respectively. Since we immediately measure the ancilla qubit after syndrome extraction, the labeling of ancilla qubits do not have conflicts. 

For individual $X$-syndrome extraction tensor at time $t$ and $Z$-syndrome extraction tensors at time $t+1$, they have gauge groups
\begin{eqs}\label{eq:SE_gauge}
    \mathcal{R}_{X-\mathrm{SE}}^{(t)}=\langle & X_{\mathrm{prep}}^{(t)} (\mathrm{e}_{N+j})X_{\mathrm{meas}}^{(t)}(\mathrm{e}_{N+j}) X_{i}^{(t)}(h_j),~~ Z_{i}^{(t)}(\mathrm{e}_v) Z_{o}^{(t)}(\mathrm{e}_v) \bigotimes_{\substack{h_f \cdot \mathrm{e}_v \neq 0}} Z_{\mathrm{meas}}^{(t)} (\mathrm{e}_{N+f}),\\
    &X_{i}^{(t)}(\mathrm{e}_v) X_{o}^{(t)}(\mathrm{e}_v),~~ Z_{\mathrm{prep}}^{(t)}(\mathrm{e}_{N+j}) Z_{\mathrm{meas}}^{(t)}(\mathrm{e}_{N+j}): \quad \forall j=1,..., r_X, ~  v=1,...,N\rangle,\\
    \mathcal{R}_{Z-\mathrm{SE}}^{(t+1)}=\langle &Z_{\mathrm{prep}}^{(t+1)}(\mathrm{e}_{N+l}) Z_{\mathrm{meas}}^{(t+1)} (\mathrm{e}_{N+l}) Z_{i}^{(t+1)} (g_l), \quad X_{i}^{(t+1)}(\mathrm{e}_v) X_{o}^{(t+1)} (\mathrm{e}_v) \bigotimes_{\substack{g_f\cdot \mathrm{e}_v\neq 0}} X_{\mathrm{meas}}^{(t+1)}(\mathrm{e}_{N+f}),\\
    &Z_{i}^{(t+1)}(\mathrm{e}_v) Z_{o}^{(t+1)} (\mathrm{e}_v),~~ X_{\mathrm{prep}}^{(t+1)} (\mathrm{e}_{N+l}) X_{\mathrm{meas}}^{(t+1)} (\mathrm{e}_{N+l}) : \quad \forall l=1,..., r_Z, \quad v=1,..., N \rangle,
\end{eqs}
where $\mathrm{e}_j$ is a one-hot binary vector that only has 1 on the $j$-th entry. Here $X_{\mathrm{prep}}^{(t)}(\mathrm{e}_{N+j})$ and $ X_{\mathrm{meas}}^{(t)}(\mathrm{e}_{N+j})$ represents the ancilla initialization and measurement when we measure a $X$-check given by row vector $h_j$ at time $t$, similar for $Z_{\mathrm{prep}}^{(t+1)}(\mathrm{e}_{N+l})$ and $ Z_{\mathrm{meas}}^{(t+1)} (\mathrm{e}_{N+l})$. The first terms in $\mathcal{R}_{X-\mathrm{SE}}^{(t)}$ and $\mathcal{R}_{Z-\mathrm{SE}}^{(t+1)}$ represent the equivalence relations between a measurement fault and the faults acting on the support of corresponding checks; the second term represents that the $Z$- and $X$-faults propagate through $X$- and $Z$-syndrome extraction respectively; the third terms correspond the fact that $X$- and $Z$-faults are free to pass through $X$-syndrome and $Z$-syndrome extraction gadgets respectively.

Then we consider the entire $\ell$-round syndrome extraction circuit, which has totally $2\ell$ timesteps, where we perform $X$-syndrome extraction in the odd timesteps and $Z$-syndrome extraction in the even timesteps. The input gauge group is
\begin{eqs}
    \mathcal{S}_\mathrm{in}=\langle X_o^{(t=0)}(h_j), Z_o^{(t=0)}(g_l), X_{\mathrm{prep}}^{(t=2m-1)} (\mathrm{e}_{N+j}), Z_{\mathrm{prep}}^{(t=2m)} (\mathrm{e}_{N+l}) :~~ \forall j=1,..., r_X, ~~ l=1,..., r_Z, ~~ m=1,..., \ell \rangle.
\end{eqs}

As shown in Fig.~\ref{fig:SE_gadget}, such circuit is formed by gluing the $X$- and $Z$-syndrome extraction tensors in an alternating pattern, which gives the contraction gauge group 
\begin{eqs}
    \mathcal{B}= \langle X_{o}^{(t)}(\mathrm{e}_{v}) X_{i}^{(t+1)}(\mathrm{e}_v), Z_{o}^{(t)}(\mathrm{e}_v) Z_{i}^{(t+1)}(\mathrm{e}_v) : \quad \forall v=1,..., N, ~~ t=0,..., 2\ell-1\rangle .
\end{eqs} 
The  measurement gauge group here is
\begin{eqs}
\mathcal{M}=\langle X_{\mathrm{meas}}^{(t=2m-1)}(\mathrm{e}_{N+j}) , Z_{\mathrm{meas}}^{(t=2m)}(\mathrm{e}_{N+l}) :\quad \forall j=1, ...., r_X, \quad l=1, ...., r_Z, \quad  m= 1,..., \ell\rangle.
\end{eqs}

After we impose $\mathcal{B},\mathcal{M}, \mathcal{S}_{\mathrm{in}}$, we can find the detectors can be written as
\begin{eqs}
    D_X&= \langle X_{\mathrm{prep}}^{(t-1)}(\mathrm{e}_{N+j}) X_{\mathrm{meas}}^{(t-1)}(\mathrm{e}_{N+j}) X_{\mathrm{prep}}^{(t+1)} (\mathrm{e}_{N+j}) X_{\mathrm{meas}}^{(t+1)}(\mathrm{e}_{N+j}) \mathbf{X}^{(t-1\rightarrow t+1)}(h_j)
    : \quad \forall  j= 1,..., r_X, \quad  t=2,\ldots,2\ell-2 \rangle, \\
    D_Z&= \langle Z_{\mathrm{prep}}^{(t)}(\mathrm{e}_{N+j}) Z_{\mathrm{meas}}^{(t)}(\mathrm{e}_{N+j}) Z_{\mathrm{prep}}^{(t+2)}(\mathrm{e}_{N+j}) Z_{\mathrm{meas}}^{(t+2)}(\mathrm{e}_{N+j}) \mathbf{Z}^{(t\rightarrow t+2)}(g_l) : \quad \forall l= 1,..., r_Z, \quad t=2,\ldots,2\ell-2 \rangle,
\end{eqs}
where 
\begin{eqs}
    \mathbf{X}^{(t-1 \rightarrow t+1)}(h_j)&:= X_{o}^{(t-1)} (h_j) X_{i}^{(t)}(h_j) X_{o}^{(t)}(h_j) X_{i}^{(t+1)} (h_j),\\
    \mathbf{Z}^{(t\rightarrow t+2)}(g_l)&:= Z_{o}^{(t)}(g_l) Z_{i}^{(t+1)}(g_l) Z_{o}^{(t+1)}(g_l) Z_{i}^{(t+2)} (g_l).
\end{eqs}
The operator
$\mathbf{X}^{(t-1\rightarrow t+1)}(h_j)$
detects $Z$-type data faults supported on $h_j$ between two consecutive rounds of $X$-syndrome extraction. Similarly, $\mathbf{Z}^{(t\rightarrow t+2)}(g_l)$ detects $X$-type data faults supported on $g_l$ between two consecutive rounds of $Z$-syndrome extraction. Note the detectors are contributed by the first, third terms of Eq.~\eqref{eq:SE_gauge} and $\mathcal{B}$, the second and fourth terms in Eq.~\eqref{eq:SE_gauge} will be completely removed by ancilla's single-qubit measurements. 

Assuming that the initializations of ancilla qubits are perfect, we can remove all the ancilla initialization fault positions. The $X$- and $Z$-detectors are simplified to
\begin{eqs}\label{eq:foliation_detectors}
    D_X&= \langle M_{h_j}^{(t-1)} M_{h_j}^{(t+1)} \mathbf{X}^{(t-1\rightarrow t+1)}(h_j)
     : \quad \forall  j= 1,..., r_X, ~~ t=2,\ldots,2\ell-2 \rangle, \\
    D_Z&= \langle M_{g_l}^{(t)}  M_{g_l}^{(t+2)} \mathbf{Z}^{(t \rightarrow t+2)}(g_l)  : \quad  \forall l= 1,..., r_Z, ~~ t=2,\ldots,2\ell-2 \rangle,
\end{eqs}
which recovers the $X$- and $Z$-detectors of the foliation \cite{faultcomplex}. Here we use $M_{h_j}^{(t\pm 1)}$ denotes the measurement fault position for measuring $X$-check $X(h_j)$ at time $t\pm 1$, similarly $ M_{g_l}^{(t)}$ denotes the measurement fault position for measuring $Z$-check $Z(g_l)$ at time $t$. Because the repeated syndrome extraction under phenomenological error model (foliation) has hypergraph product structure, when we consider a $\ell$-round syndrome extraction gadget together with minimum-weight recovery, the bulk effective distance is $d_{\mathrm{bulk}}=\lfloor \frac{1}{2}(\min\{\ell,d\}-1) \rfloor$ meaning that every fault with weight up to $d_{\mathrm{bulk}}$ is correctable.

According to Lemma~\ref{lemma:bulk_detectable_distance},  every bulk fault configuration $F$ satisfying $|F|<2d_{\mathrm{bulk}}+1=\min \{\ell,d\}$
is either detectable or has trivial logical action. 
Thus, fault tolerance in the bulk is limited by the smaller of the spatial code distance $d$ and the temporal fault distance $\ell$. This agrees with the usual expectation that at least $d$ rounds of syndrome extraction are required before minimum-weight recovery to make both bulk fault distance and spacetime fault distance achieve $d$, and hence to realize a fully fault-tolerant syndrome extraction gadget satisfying the ECCP stated in Theorem~\ref{thm:ECCP}. By contrast, the fault distance of the full spacetime complex is $d$. This distinction arises from the input and output boundary relations, which render both temporal boundaries smooth. More precisely, for each code stabilizer $S_i$, there is a spacetime stabilizer supported on $Q_{\mathrm{input}}$ and the corresponding measurement port in the first round of syndrome extraction, together with a final boundary stabilizer relation supported on $Q_{\mathrm{output}}$ and the corresponding measurement port in the last round of syndrome extraction. These boundary terms prevent the existence of temporal undetectable fault crossing the time. The full spacetime fault distance is therefore determined solely by the code distance. 

\section{Fault tolerance of logical gate implementations}

In this section, we analyze the correlated detectors that cross transversal CNOT and fold-transversal gates, and show that including these detectors naturally takes the error propagation into account and preserves the fault distance. Specifically, we show that if we insert a transversal gate inside a syndrome extraction gadget under phenomenological error model, the fault distance is unchanged. 

\subsection{Transversal CNOT gate}

In this section, we show that inserting a transversal CNOT gate into a repeated syndrome extraction gadget of a CSS code, constructed from the $X$- and $Z$-syndrome extraction tensors in Eq.~\eqref{eq:SE_gauge}, preserves the fault distance. More precisely, if the original syndrome extraction $\mathrm{SE}_a\circ\mathrm{SE}_b$ has fault distance $d$, then the modified fault-tolerant protocol $\mathrm{SE}_a\circ\mathrm{CNOT}^{\otimes N}\circ\mathrm{SE}_b
$ also has fault distance $d$. We insert the transversal CNOT at time $t+1$ where $X$ syndrome extraction is implemented at times $t$ and $t+3$, whereas the $Z$ syndrome extraction is implemented at times $t-1$ and $t+2$. For the rest of time $X$ and $Z$ syndrome extractions happen alternately.



We work under a phenomenological noise model in which elementary faults are single-qubit Pauli data errors or measurement-outcome flips.  Since inserting a transversal CNOT gate only changes the detectors on the interface of transversal CNOT gate and the timelike undetectable fault cannot start and terminate at the interfaces, so we only need to prove the spacelike fault distance at the interfaces of transversal CNOT gate is $d$. We first consider data faults occurring immediately before and after the CNOT. Let $P_1$ and $P_2$ denote the corresponding Pauli operators. Propagating the earlier faults through the gate gives the error $P'=P_2UP_1U^\dagger$ happened on the input at time $t+2$, as illustrated Fig.~\ref{fig:CNOT_propagate}.

\begin{figure}[t]
    \centering
    \includegraphics[width=0.7\textwidth]{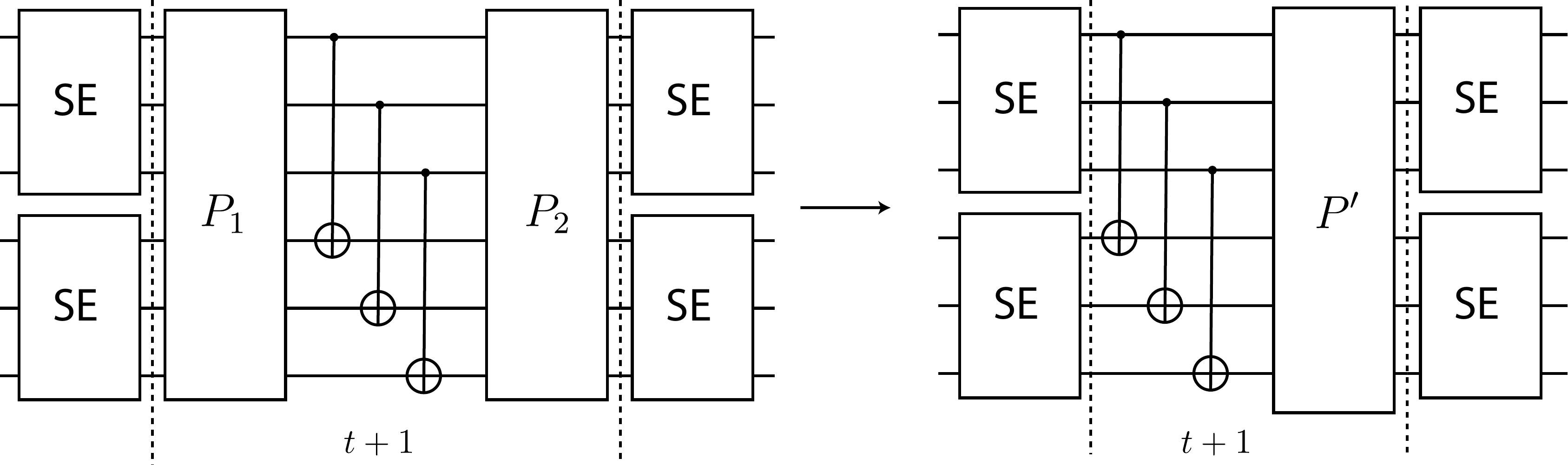}
    \caption{We consider a faulty transversal CNOT gate between two syndrome extraction (SE) gadget whose fault distances are both $d$. Suppose we have weight-$r$ Pauli faults $P_1, P_2$ happening right before and after transversal CNOT gate. 
    By applying the gauge generators of the transversal CNOT tensor, we can propagate all faults to immediately after the gate.
    After we do so, the fault $P'$ can be written as $P'=X_i^{(t+2)}(p) X_i^{(t+2)}(p_1) Z_i^{(t+2)}(q) Z_i^{(t+2)}(q_1) \otimes X_i^{(t+2)}(p) X_i^{(t+2)}(p_2) Z_i^{(t+2)}(q) Z_i^{(t+2)}(q_2) $ acting on two code blocks, where we decompose the fault into correlated part $X_i^{(t+2)}(p) Z_i^{(t+2)}(q) \otimes X_i^{(t+2)}(p) Z_i^{(t+2)}(q)$ and uncorrelated part $X_i^{(t+2)}(p_1) Z_i^{(t+2)}(q_1) \otimes X_i^{(t+2)}(p_2) Z_i^{(t+2)}(q_2)$ .The weights of faults satisfy $|p|+|p_1|+|p_2|\leq r, |q|+|q_1|+|q_2|\leq r$.}
    \label{fig:CNOT_propagate} 
\end{figure}

Up to an overall phase, the \(X\)- and \(Z\)-components of \(P'\) can be written as
\begin{eqs}
P'_X &=
X_i^{(t+2)}(p+p_1)\otimes X_i^{(t+2)}(p+p_2),\\
P'_Z
&=
Z_i^{(t+2)}(q+q_1)\otimes Z_i^{(t+2)}(q+q_2).
\end{eqs}

Here, \(p\) and \(q\) describe the error components correlated between the two codeblocks, while \(p_1,p_2\) and \(q_1,q_2\) describe the uncorrelated components. These vectors can be chosen to have pairwise-disjoint supports:
\begin{eqs}
    \operatorname{supp}(p)\cap\operatorname{supp}(p_1)
=
\operatorname{supp}(p)\cap\operatorname{supp}(p_2)
=
\operatorname{supp}(p_1)\cap\operatorname{supp}(p_2)
=
\varnothing,
\end{eqs}
and similarly
\begin{eqs}
\operatorname{supp}(q)\cap\operatorname{supp}(q_1)
=
\operatorname{supp}(q)\cap\operatorname{supp}(q_2)
=
\operatorname{supp}(q_1)\cap\operatorname{supp}(q_2)
=
\varnothing.
\end{eqs}

\begin{figure}[t]
    \centering
    \includegraphics[width=0.7\textwidth]{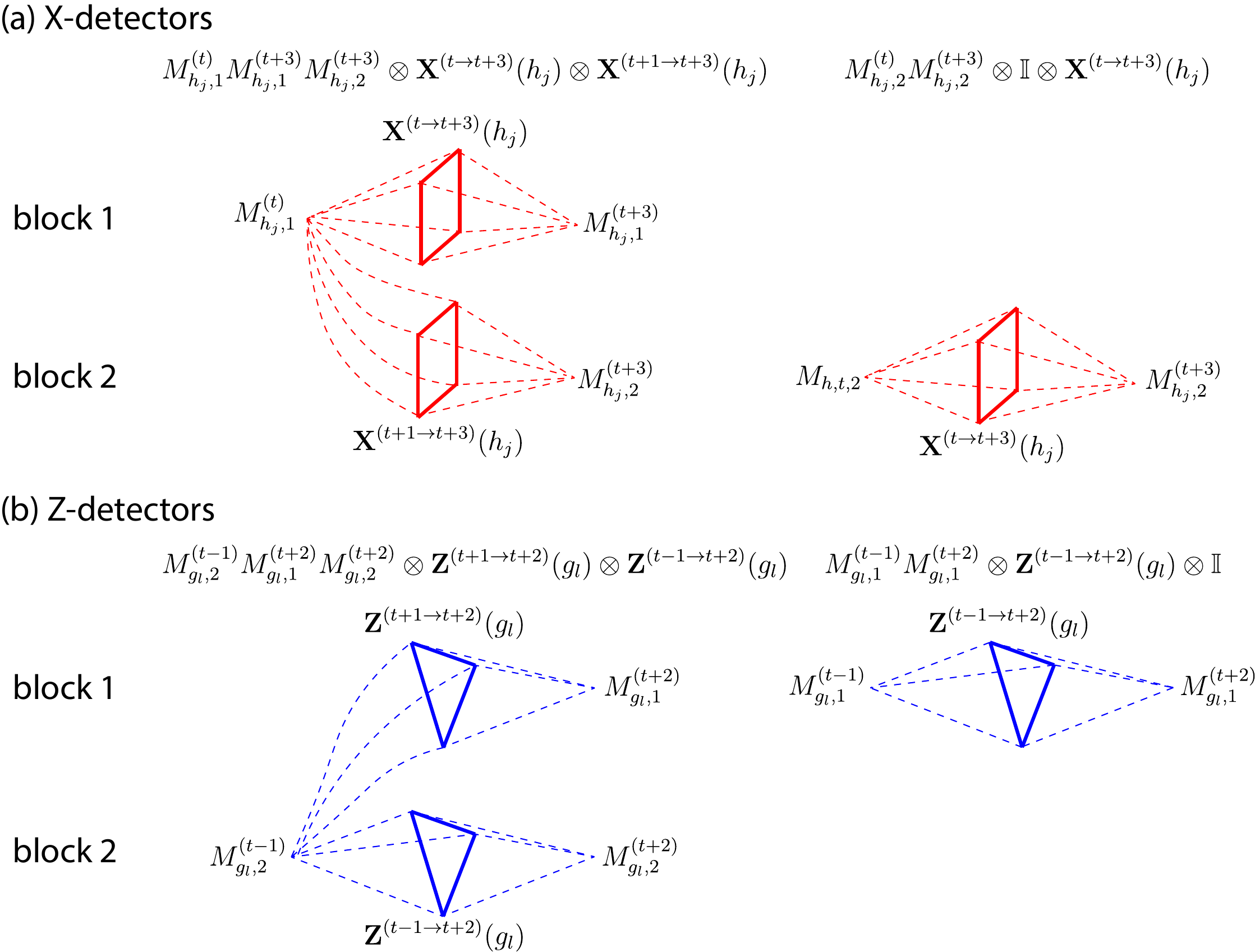}
\caption{$X$- and $Z$-detectors crossing the transversal CNOT gate in Eq.~\eqref{eq:correlated_detectors}. }
    \label{fig:CNOT_detectors}
\end{figure}

Consequently, the numbers of affected locations after transversal CNOT gate satisfy
\begin{eqs}
    |p|+|p_1|+|p_2|\leq r,
\qquad
|q|+|q_1|+|q_2|\leq r.
\end{eqs}

The $X$- and $Z$-detectors crossing the transversal CNOT are
\begin{eqs}\label{eq:correlated_detectors}
    D_{X,\mathrm{CNOT}}=&\langle M_{h_j,1}^{(t)} M_{h_j,1}^{(t+3)}  M_{h_j,2}^{(t+3)}  \otimes \mathbf{X}^{(t\rightarrow t+3)}(h_j) \otimes  \mathbf{X}^{(t+1\rightarrow t+3)}(h_j), ~~M_{h_j,2}^{(t)} M_{h_j,2}^{(t+3)}\otimes  \mathbb{I} \otimes \mathbf{X}^{(t\rightarrow t+3)}(h_j) : \quad \forall j= 1,..., r_X \rangle,\\
    D_{Z,\mathrm{CNOT}}=&\langle M_{g_l,2}^{(t-1)} M_{g_l,1}^{(t+2)} M_{g_l,2}^{(t+2)}\otimes  \mathbf{Z}^{(t+1\rightarrow t+2)}(g_l)\otimes \mathbf{Z}^{(t-1\rightarrow t+2)}(g_l),~~M_{g_l,1}^{(t-1)} M_{g_l,1}^{(t+2)} \otimes \mathbf{Z}^{(t-1\rightarrow t+2)}(g_l) \otimes \mathbb{I}  : \quad \forall l= 1,..., r_Z \rangle,
\end{eqs}
where 
\begin{eqs}
    \mathbf{X}^{(t_1\rightarrow t_2)}(h_j)&:= X_o^{(t_1)}(h_j) X_i^{(t_1+1)}(h_j)\cdots X_i^{(t_2)}(h_j) ,\\
    \mathbf{Z}^{(t_1\rightarrow t_2)}(g_l)&:= Z_o^{(t_1)}(g_l) Z_i^{(t_1+1)}(g_l) \cdots Z_i^{(t_2)}(g_l)
\end{eqs}
denotes the detecting region on the data fault between time $t_1 $ and $t_2$.

Here, \(M_{h_j,1}^{(t)},M_{h_j,2}^{(t)}\) denotes a measurement fault positions on the ancilla used to measure the \(X\)-type check $X(h_j)$ on first and second code block at time \(t\). The notation \(M_{g_l,1}^{(t)}, M_{g_l,2}^{(t)}\) is defined analogously for the measurement of the \(Z\)-type check \(Z(g_l)\) on the first and second block at time $t$.

The resulting detectors couple the two code blocks, as shown in Fig.~\(\ref{fig:CNOT_detectors}\). Away from the CNOT layer, the detectors remain blockwise independent.

Jointly decoding the two code blocks using these cross-block detectors accounts for the propagation of Pauli errors through the transversal CNOT. This procedure, commonly called correlated decoding, can reduce the logical error rate and increase the threshold relative to decoding the two blocks independently using the detectors in Eq.~\eqref{eq:foliation_detectors}~\cite{correlated_decoding,zhou2025low}. Importantly, correlated decoding modifies only the detector matrices supplied to the classical decoder; the underlying syndrome extraction circuit remains unchanged from that used for uncorrelated decoding.

For example, consider an error \(X(u)\otimes X(u+v)\) occurring immediately after the transversal CNOT, where$X(u+v)=\overline X$ is a logical operator and $X(u)$ has a nontrivial syndrome. If the two blocks are decoded independently, the error on the second block has a trivial syndrome and is interpreted as a logical error. When we decode using the cross-block detectors, however, this error can be detected by $D_{Z,\mathrm{CNOT}}$ after taking the error propagation across transversal CNOT gate into account.

We first consider the $X$-component of the propagated fault $X_i^{(t+2)}(e_1) \otimes X_i^{(t+2)}(e_2)$ where $e_1=p+p_1, e_2= p+p_2$. The detector $M_{g_l,1}^{(t-1)} M_{g_l,1}^{(t+2)} \otimes \mathbf{Z}^{(t-1\rightarrow t+2)}(g_l) \otimes \mathbb{I}
$ measures the syndrome of $e_1$, whereas
$M_{g_l,2}^{(t-1)} M_{g_l,1}^{(t+2)} M_{g_l,2}^{(t+2)}\otimes  \mathbf{Z}^{(t+1\rightarrow t+2)}(g_l)\otimes \mathbf{Z}^{(t-1\rightarrow t+2)}(g_l)$
 measures the syndrome of $e_1+e_2$. If the error is undetected, then
\begin{eqs}
    H_Z e_1=0,
\qquad
H_Z(e_1+e_2)=0,
\end{eqs}
hence both $e_1$ and $e_2$ have trivial $Z$-syndrome. Since $|e_1|,|e_2|<d$, neither $X(e_1)$ nor $X(e_2)$ can be a nontrivial logical operator. Thus, every $X$-fault of weight less than $d$ is either detectable or logically trivial.

Similarly, we write the \(Z\)-component of the propagated fault as $Z_i^{(t+2)}(f_1)\otimes Z_i^{(t+2)}(f_2)$ with  $f_1=q+q_1$ and $ f_2=q+q_2$.
The detectors $M_{h_j,2}^{(t)} M_{h_j,2}^{(t+3)}\otimes  \mathbb{I} \otimes \mathbf{X}^{(t\rightarrow t+3)}(h_j) 
$ and
$
M_{h_j,1}^{(t)} M_{h_j,1}^{(t+3)}  M_{h_j,2}^{(t+3)}  \otimes \mathbf{X}^{(t\rightarrow t+3)}(h_j) \otimes  \mathbf{X}^{(t+1\rightarrow t+3)}(h_j)$
measure the syndromes of $f_2$ and $f_1+f_2$, respectively. Therefore, if the error is undetected, both $f_1$ and $f_2$ have trivial $X$-syndrome. Since $|f_1|,|f_2|<d$, neither component can be a nontrivial logical operator. Hence, the spacelike fault distance between times $t$ and $t+1$ is still $d$. The spacelike distances at all other time slices remain unchanged, since their detector structure is unaffected by the insertion of the CNOT. Consequently, the spacelike fault distance is $d_{\mathrm{space}}(\mathrm{SE}_a\circ\mathrm{CNOT}^{\otimes N}\circ\mathrm{SE}_b
)=d$. 

The timelike fault distance after we insert the transversal CNOT gate is also unchanged, because an undetectable timelike fault chain cannot begin or terminate at the CNOT interface.
Consequently, inserting the transversal CNOT does not create a shorter timelike logical fault. Provided that all initial and final temporal-boundary detectors are retained, the timelike fault distance is inherited from the original syndrome extraction gadgets: $d_{\mathrm{time}}(\mathrm{SE}_a\circ\mathrm{CNOT}^{\otimes N}\circ\mathrm{SE}_b
)
=
d_{\mathrm{time}}(\mathrm{SE}_a\circ \mathrm{SE}_b
)$.

The fault distance of the spacetime circuit is determined by both its space-like and time-like fault distances. For the transversal-CNOT protocol, we have shown that the spacelike and timelike fault distances are both $d$. Given the fault distance is governed by the smaller one of spacelike and timelike fault distances, we have
\begin{eqs}
    d_\st (\mathrm{SE}_a\circ\mathrm{CNOT}^{\otimes N}\circ\mathrm{SE}_b
)=
\min\left\{
d_{\mathrm{space}}(\mathrm{SE}_a\circ\mathrm{CNOT}^{\otimes N}\circ\mathrm{SE}_b
),
d_{\mathrm{time}}(\mathrm{SE}_a\circ\mathrm{CNOT}^{\otimes N}\circ\mathrm{SE}_b
)
\right\} 
=d.
\end{eqs}

\subsection{Fold-transversal gates}

Fold-transversal gates can provide logical Clifford operations beyond those available through strictly transversal implementations~\cite{moussa_fold,breuckmann2024fold}. There are two main classes: Hadamard-type and phase-type gates. A Hadamard-type gate consists of single-qubit Hadamard gates together with a permutation of the physical qubits. 

Here, we focus on the phase-type fold-transversal gate \(S_\tau\), because the Hadamard-type fold-transversal gate does not propagate faults under the phenomenological error model. We show that inserting \(S_\tau\) into a repeated syndrome extraction gadget of a CSS code preserves its fault distance. More precisely, if 
$\mathrm{SE}_a\circ\mathrm{SE}_b $ has fault distance $d$, then the modified protocol $\mathrm{SE}_a\circ S_\tau\circ\mathrm{SE}_b$ has fault distance $d$, provided that the correlated detectors crossing $S_\tau$ are included. Similar to the setting of analysis of transversal CNOT gate, we insert $S_\tau $ at time $t+1$ where $X$ syndrome extraction is implemented at times $t$ and $t+3$, whereas the $Z$ syndrome extraction is implemented at times $t-1$ and $t+2$. For the rest of time $X$ and $Z$ syndrome extractions happen alternately.

Consider a $[\![2N,k,d]\!]$ CSS code with an permutation of qubits
\begin{eqs}
    \tau: \mathbb{F}_2^{2N} \rightarrow   \mathbb{F}_2^{2N},\qquad \tau^2=\mathrm{id}.
\end{eqs}

The permutation \(\tau\) is a \(ZX\)-duality if it exchanges the \(X\)- and \(Z\)-check spaces:
\begin{eqs}
\tau \left(\operatorname{row}(H_X)\right)
=\operatorname{row}(H_Z),
\qquad
\tau \left(\operatorname{row}(H_Z)\right)
=
\operatorname{row}(H_X).
\end{eqs}
In other words, for every $X$-check given by row vector $h_j$, it is related to a $Z$-check given by $\tau(h_j) \in \mathrm{row}(H_Z)$ and vice versa. For convenience, we pick the stabilizer generators such that their check matrices satisfy $\tau(H_X)=H_Z$.

Suppose that \(\tau\) has \(2w\) fixed points and \(n\) two-element orbits. The total number of physical qubits is then $2N=2n+2w=2(n+w)$.
For a phase-type gate to exist, the fixed-point set must admit a bipartition $Q_{\mathrm{fix}}=A\sqcup B$, such that every \(X\)-type check has equal overlap with \(A\) and \(B\). In addition, every \(X\)-type check must contain an even number of complete two-element orbits of \(\tau\). Under these conditions, the phase-type fold-transversal gate is

\begin{eqs}S_\tau
=
\left(\bigotimes_{j\in A}S_j\right)
\left(\bigotimes_{j\in B}S_j^\dagger\right)
\left(
\bigotimes_{\substack{i<\tau(i)}}
\operatorname{CZ}_{i,\tau(i)}
\right).
\end{eqs}

The detectors crossing a phase-type fold-transversal gate can be written as 
\begin{eqs}\label{eq:fold_transversal_detectors}
    D_{S_\tau}=&\langle M_{h_j}^{(t)} M_{h_j}^{(t+3)} M_{\tau(h_j)}^{(t+2)} \otimes \mathbf{X}^{(t\rightarrow t+3)}(h_j) \mathbf{Z}^{(t+1 \rightarrow t+2)}(\tau(h_j)), ~~M_{g_l}^{(t-1)} M_{g_l}^{(t+2)} \otimes  \mathbf{Z}^{(t-1 \rightarrow t+2)}(g_l) : \quad \forall j=1,...,r_X,~~ l=1,...,r_Z \rangle.
\end{eqs}
The first term gives  mixed-type detectors. It arise because conjugating an \(X\)-type check through \(S_\tau\) produces both \(X\)- and \(Z\)-components. The second term gives \(Z\)-type detectors, since \(S_\tau\) leaves \(Z\)-type operators unchanged. Consequently, the detector group crossing \(S_\tau\) is no longer CSS.

By dividing the system into two folded sectors, a weight-$r$ fault can be written as a sequential action of fault before and after the fold-transversal gate. We write the fault occurring before and after $S_\tau$ as  $P_1=X_o^{(t)}(p_1) Z_o^{(t)}(p_2)$ and $P_2=X_i^{(t+2)}(p_3) Z_i^{(t+2)}(p_4)$, respectively. We use the gauge checks to push all the fault to after the fold-transversal $S_\tau$ gate, which yields an equivalent fault 
\begin{eqs}
    P_2 S_\tau P_1 S_\tau^\dagger= X_i^{(t+2)}(x)Z_i^{(t+2)}(z),
\end{eqs}
where 
\begin{eqs}
    x= p_1 + p_3,\quad z= p_2 + p_4 +\tau (p_1).
\end{eqs}

Here we assume that each elementary data fault is a single-qubit $X$ or $Z$ error, and the weight counts these events additively. Therefore, given $|p_1|+|p_2|+ |p_3| + |p_4|\leq r$, we have $|x|, |z|\leq r$. 

After setting the measurement fault variables to the identity, the relevant detector operators are
\begin{eqs}
    X_i^{(t+2)}(h_j)\otimes Z_i^{(t+2)}(\tau(h_j)),
\quad
Z_i^{(t+2)}(g_l),
\end{eqs}
for $j=1,...,r_X$ and $l=1,...,r_Z$.

If $X_i^{(t+2)}(x)Z_i^{(t+2)}(z)$ is undetected, commutation with detectors gives
\begin{eqs}
H_X z+ H_Z x=0,\quad 
H_Z x=0.
\end{eqs}
Above equations are equivalent to 
\begin{eqs}
    H_Z x=0,\quad H_X z=0.
\end{eqs} 
Hence, both $X_i^{(t+2)}(x)$ and $Z_i^{(t+2)}(z)$  have trivial stabilizer syndrome.
If $r<d$, then $|x|,|z|< d$. Given the input code with check matrices $H_X, H_Z$ has distance $d$, none of these components can be a nontrivial logical operator. Each component is therefore either the identity or a stabilizer, so $X_i^{(t+2)}(x)Z_i^{(t+2)}(z)$ acts trivially on the encoded information. we conclude spacelike fault distance at the interface of phase-type fold-transversal gate is still $d$.

For measurement fault, the argument is similar to the one used in transversal gates, because an undetectable timelike fault chain cannot begin or terminate on the interface of fold-transversal gate. So we have $d_{\mathrm{time}}(\mathrm{SE}_a \circ S_\tau \circ \mathrm{SE}_b)
=
d_{\mathrm{time}}(\mathrm{SE}_a  \circ \mathrm{SE}_b)$.

Therefore, the fault distance is
\begin{eqs}
d_\st (\mathrm{SE}_a  \circ S_\tau \circ \mathrm{SE}_b)
=
\min\left\{
d_{\mathrm{space}}(\mathrm{SE}_a \circ S_\tau \circ \mathrm{SE}_b),
d_{\mathrm{time}}(\mathrm{SE}_a \circ S_\tau \circ \mathrm{SE}_b)
\right\}
=d,
\end{eqs}
which is unchanged after inserting the fold-transversal gate $S_\tau$.

\section{Fault tolerance of various error correction schemes}

\subsection{Steane error-correction gadget}
\label{sec:steane}

Consider an $[\![n,k,d]\!]$ CSS stabilizer code. The $n$ physical qubits of the data block are labeled $1,\ldots,n$. The $n$ physical qubits of ancilla block $a$ and ancilla block $b$ are each labeled $1,\ldots,n$, respectively. We denote the one-hot vector selecting the $q$th qubit of any block by $e_q\in\F^n$.

The original CSS code stabilizer generators are written as
\begin{equation}
\St_{\rm CSS}=\left\langle X(h_j),Z(g_l):
\quad \forall j=1,\ldots,r_X,\quad l=1,\ldots,r_Z \right\rangle.
\label{eq:cssgens}
\end{equation}
with
\begin{equation}
C_X=\Span\{h_j\}_{j=1}^{r_X},\quad
C_Z=\Span\{g_l\}_{l=1}^{r_Z},\quad
C_X\subseteq C_Z^\perp,\quad r_X+r_Z=n-k.
\label{eq:cssspaces}
\end{equation}
Here $X(v)=\bigotimes_{q=1}^nX_q^{v_q}$ and $Z(v)=\bigotimes_{q=1}^nZ_q^{v_q}$ for $v\in\F^n$. Choose pure logical representatives $X(\ell_{X,\mu})$ and $Z(\ell_{Z,\mu})$, $\mu=1,\ldots,k$, such that
\begin{equation}
\begin{aligned}
C_Z^\perp&=C_X\oplus\Span\{\ell_{X,\mu}\}_{\mu=1}^k,
& C_X^\perp&=C_Z\oplus\Span\{\ell_{Z,\mu}\}_{\mu=1}^k,&
\ell_{X,\mu}\cdot\ell_{Z,\nu}=\delta_{\mu\nu}.&&
\end{aligned}
\label{eq:csslogicals}
\end{equation}

The gadget has three time steps. At $t=0$, the data block is in an unknown encoded logical state and block $a$ is prepared in $\ket{\overline{+^k}}$. At $t=1$, a transversal $\mathrm{CNOT}_{d\to a}$ is applied; immediately on the same output legs, block $a$ is measured qubitwise in the $Z$ basis. Also at $t=1$, block $b$ is prepared in $\ket{\overline{0^k}}$. At $t=2$, a transversal $\mathrm{CNOT}_{b\to d}$ is applied and block $b$ is measured immediately on its output legs in the $X$ basis.

\begin{figure}[t]
    \centering
    \includegraphics[width=0.45\textwidth]{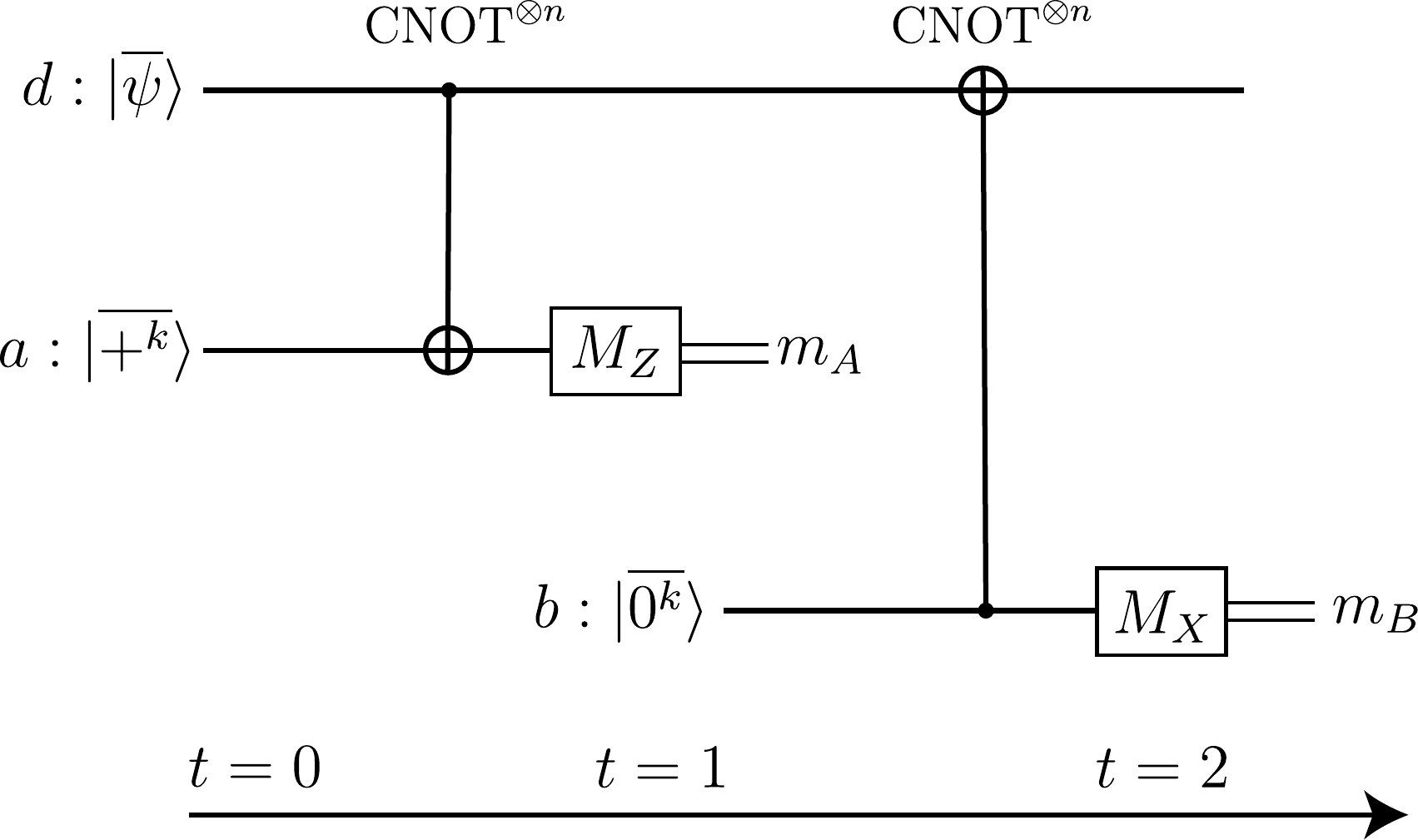}
    \caption{General Steane extraction gadget. Each line denotes an $n$-qubit block, each controlled symbol denotes $n$ parallel physical CNOTs, and the measurement is performed directly on the output leg of the preceding gate layer. Block $b$ is prepared on the $t=1$ output qubits.}
\label{fig:steane}
\end{figure}

For all calculations, we take Pauli operators modulo scalar phases. An $N$-qubit Pauli has binary vector $e=(x\mid z)\in\F^{2N}$ with symplectic form
\begin{equation}
\lambda(e,f)=x\cdot z'+z\cdot x',\qquad f=(x'\mid z').
\label{eq:symplectic}
\end{equation}
For a Pauli subspace $\G$, its centralizer is $\G^\perp$ and its stabilizer subspace is $\St=\G\cap\G^\perp$. Bare and dressed logical spaces are $\G^\perp/\St$ and $\St^\perp/\G$, respectively~\cite{Poulin2005}. The spacetime distance is
\begin{equation}
 d_{\st}=\min_{\mathbf{e}\in\St^\perp \backslash \G}|\mathbf{e}|
 =\min_{[\mathbf{e}]\in\St^\perp/\G \backslash \{0\}}
       \min_{\mathbf{g}\in\G}|\mathbf{e+g}|.
\label{eq:distance}
\end{equation}

In summary, the spacetime code for the Steane EC gadget contains the blocks
\begin{equation}
\begin{array}{c|l}
0&d_o^{(0)},\ a_o^{(0)},\\
1&d_i^{(1)},\ a_i^{(1)},\ d_o^{(1)},\ a_o^{(1)},\ b_o^{(1)},\\
2&d_i^{(2)},\ b_i^{(2)},\ d_o^{(2)},\ b_o^{(2)}.
\end{array}
\label{eq:steaneblocks}
\end{equation}
Every listed block contains $n$ qubits, so the spacetime code acts on $N_{\rm S}=11n$ physical spacetime qubits.

The following 
\label{subsec:steanegauge}
Equations~\eqref{eq:steaneprep}--\eqref{eq:steanemeasureA} generate the gauge of the spacetime code $\mathcal{G}=\langle \mathcal{S}_{\mathrm{in}},\mathcal{B},\mathcal{R},\mathcal{M}\rangle$.
The preparation gauge set $\mathcal{S}_{\mathrm{in}}$ includes
\begin{align}
\mathcal{S}_{\mathrm{in}}=\langle &X_{d_o}^{(0)}(h_j),\quad Z_{d_o}^{(0)}(g_l),\quad X_{a_o}^{(0)}(h_j),\quad Z_{a_o}^{(0)}(g_l),\quad
 X_{a_o}^{(0)}(\ell_{X,\mu}),\quad X_{b_o}^{(1)}(h_j),\quad Z_{b_o}^{(1)}(g_l),\quad
 Z_{b_o}^{(1)}(\ell_{Z,\mu}):\\
 &\forall j=1,\ldots,r_X,~l=1,\ldots,r_Z, ~\mu=1,\ldots,k\rangle.
\label{eq:steaneprep}
\end{align}
The first two groups fix the input code space but do not specify the logical state; the remaining groups of generators prepare the ancilla blocks on $\ket{\overline{+^k}}_A$ and $\ket{\overline{0^k}}_B$, respectively.

The bond gauges are
\begin{align}
\mathcal{B}=\langle&X_{d_o}^{(0)}(e_q)X_{d_i}^{(1)}(e_q),\quad
 Z_{d_o}^{(0)}(e_q)Z_{d_i}^{(1)}(e_q),\quad X_{a_o}^{(0)}(e_q)X_{a_i}^{(1)}(e_q),\quad
 Z_{a_o}^{(0)}(e_q)Z_{a_i}^{(1)}(e_q),\notag\\
&X_{d_o}^{(1)}(e_q)X_{d_i}^{(2)}(e_q),\quad
 Z_{d_o}^{(1)}(e_q)Z_{d_i}^{(2)}(e_q),\quad X_{b_o}^{(1)}(e_q)X_{b_i}^{(2)}(e_q),\quad
 Z_{b_o}^{(1)}(e_q)Z_{b_i}^{(2)}(e_q):~ \forall q=1,\ldots,n\rangle.
\label{eq:steanewires}
\end{align}
The gate gauge generators are 

\begin{align}
\mathcal{R}=\langle& X_{d_i}^{(1)}(e_q)X_{d_o}^{(1)}(e_q)X_{a_o}^{(1)}(e_q),\quad Z_{d_i}^{(1)}(e_q)Z_{d_o}^{(1)}(e_q),
\quad X_{a_i}^{(1)}(e_q)X_{a_o}^{(1)}(e_q),\quad Z_{a_i}^{(1)}(e_q)Z_{d_o}^{(1)}(e_q)Z_{a_o}^{(1)}(e_q),\\
&X_{b_i}^{(2)}(e_q)X_{b_o}^{(2)}(e_q)X_{d_o}^{(2)}(e_q),\quad Z_{b_i}^{(2)}(e_q)Z_{b_o}^{(2)}(e_q),\quad X_{d_i}^{(2)}(e_q)X_{d_o}^{(2)}(e_q),\quad Z_{d_i}^{(2)}(e_q)Z_{b_o}^{(2)}(e_q)Z_{d_o}^{(2)}(e_q):~\forall q=1,\ldots,n\rangle.
\label{eq:steanegateA}
\end{align}
in which the first transversal $\mathrm{CNOT}_{d\to a}$ and the second transversal $\mathrm{CNOT}_{b\to d}$ correspond to the first and second line of the gauge generators, respectively.

The measurement gauges are 
\begin{equation}
\mathcal{M}=\langle Z_{a_o}^{(1)}(e_q),\quad X_{b_o}^{(2)}(e_q): \quad \forall q=1,\ldots,n\rangle.
\label{eq:steanemeasureA}
\end{equation}
They correspond to the $Z$-basis measurement of block $a$ and the $X$-basis measurement of block $b$.

We can calculate the generating set for the stabilizer center as follows,
\begin{equation}
\begin{aligned}
S^X_{\co,j}&=
X_{d_o}^{(0)}(h_j)X_{d_i}^{(1)}(h_j) X_{d_o}^{(1)}(h_j)X_{d_i}^{(2)}(h_j)X_{d_o}^{(2)}(h_j)X_{a_o}^{(0)}(h_j)X_{a_i}^{(1)}(h_j),\quad j=1,\ldots,r_X,\\
S^Z_{\co,l}&=
Z_{d_o}^{(0)}(g_l)Z_{d_i}^{(1)}(g_l) Z_{d_o}^{(1)}(g_l)Z_{d_i}^{(2)}(g_l)Z_{d_o}^{(2)}(g_l)Z_{b_o}^{(1)}(g_l)Z_{b_i}^{(2)}(g_l),\quad l=1,\ldots,r_Z,\\
S^A_{\de,l}&=Z_{d_o}^{(0)}(g_l)Z_{d_i}^{(1)}(g_l)
Z_{a_o}^{(0)}(g_l)Z_{a_i}^{(1)}(g_l)Z_{a_o}^{(1)}(g_l),\quad l=1,\ldots,r_Z,\\
S^B_{\de,j}&=
X_{d_o}^{(0)}(h_j)X_{d_i}^{(1)}(h_j)  X_{d_o}^{(1)}(h_j)X_{d_i}^{(2)}(h_j)X_{a_o}^{(0)}(h_j)X_{a_i}^{(1)}(h_j)X_{b_o}^{(1)}(h_j)X_{b_i}^{(2)}(h_j)
X_{b_o}^{(2)}(h_j),\quad j=1,\ldots,r_X.
\end{aligned}
\label{eq:steanestabilizers}
\end{equation}

Due to the existence of detectors, the first two groups in Eq.~\eqref{eq:steaneprep} are redundant in the generated gauge subspace. Thus, even if the input state is unrestricted and the corresponding stabilizers are omitted, the algebraic structure of the spacetime code is unchanged.

In summary, the spacetime code associated with the Steane EC gadget is a subsystem code with parameters
\begin{equation}
N=11n,\quad \dim\G=20n,\quad
\dim\St=2(n-k),\quad K=k,\quad R=9n+k.
\label{eq:steaneranks}
\end{equation}

\begin{proposition}[fault distance of the Steane EC gadget]
The spacetime subsystem code associated with the Steane error correction gadget is a subsystem code with 
\begin{equation}
[\![N,K,R,d_{\mathrm{st}}]\!]=[\![11n,k,9n+k,d]\!].
\label{eq:steaneparameters}
\end{equation}
\end{proposition}

\begin{proof}
For a spacetime error $E$, let $x_{d_o}^{(0)}(E)$, $z_{d_o}^{(0)}(E)$, etc., denote its binary $X$ and $Z$ support vectors on the indicated block. Define
\begin{align}
q(E)={}&x_{d_o}^{(0)}+x_{d_i}^{(1)}+x_{d_o}^{(1)}
+x_{d_i}^{(2)}+x_{d_o}^{(2)}+x_{b_o}^{(1)}+x_{b_i}^{(2)},\notag\\
p(E)={}&z_{d_o}^{(0)}+z_{d_i}^{(1)}+z_{d_o}^{(1)}
+z_{d_i}^{(2)}+z_{d_o}^{(2)}+z_{a_o}^{(0)}+z_{a_i}^{(1)},\notag\\
\rho_{\rm S}(E)={}&(q(E)\mid p(E)),
\label{eq:steanerho}
\end{align}
where the explicit argument $(E)$ is suppressed on the right-hand side. Also define
\begin{align}
\eta_A(E)={}&x_{d_o}^{(0)}+x_{d_i}^{(1)}+x_{a_o}^{(0)}
+x_{a_i}^{(1)}+x_{a_o}^{(1)},\notag\\
\eta_B(E)={}&z_{d_o}^{(0)}+z_{d_i}^{(1)}+z_{d_o}^{(1)}
+z_{d_i}^{(2)}+z_{a_o}^{(0)}+z_{a_i}^{(1)}+z_{b_o}^{(1)}+z_{b_i}^{(2)}+z_{b_o}^{(2)}.
\label{eq:steaneeta}
\end{align}
The complete stabilizer syndrome is
\begin{equation}
\begin{aligned}
\lambda(E,S^X_{\co,j})&=h_j\cdot p(E),
&\lambda(E,S^Z_{\co,l})&=g_l\cdot q(E),\\[-1pt]
\lambda(E,S^A_{\de,l})&=g_l\cdot\eta_A(E),
&\lambda(E,S^B_{\de,j})&=h_j\cdot\eta_B(E).
\end{aligned}
\label{eq:steanesyndrome}
\end{equation}
Hence, a zero spacetime stabilizer syndrome implies
\begin{equation}
q(E)\in C_Z^\perp,\qquad p(E)\in C_X^\perp.
\label{eq:steanenormalizer}
\end{equation}
Let
\begin{equation}
W_{\rm CSS}=\{(x\mid z):x\in C_X,\ z\in C_Z\}
\label{eq:cssW}
\end{equation}
be the original CSS stabilizer subspace. Multiplication by any gauge generator changes $\rho_{\rm S}(E)$ by an element of $W_{\rm CSS}$; therefore $[\rho_{\rm S}(E)]_{W_{\rm CSS}}$ is a well-defined function of the gauge class.

For a zero-syndrome $E$, the final-data Pauli $P_{d_o}^{(2)}(\rho_{\rm S}(E))$ has the same full stabilizer syndrome and the same commutation bits with the bare logical basis as $E$, according to Corollary \ref{coro:singlemeasreduc}. Their difference therefore lies in $\G_{\rm S}$, and
\begin{equation}
E\sim_{\G}P_{d_o}^{(2)}(\rho_{\rm S}(E)),\qquad
\St^\perp/\G\simeq W_{\rm CSS}^\perp/W_{\rm CSS}.
\label{eq:steanequotient}
\end{equation}

For a coset $c=v+W_{\rm CSS}$, let
\begin{equation}
d_c=\min_{w\in W_{\rm CSS}} |v+w|.
\label{eq:classdistance}
\end{equation}
Each physical coordinate of $\rho_{\rm S}(E)$ receives contributions only from spacetime qubits with the same coordinate in different blocks. Hence, a nonzero coordinate of $\rho_{\rm S}(E)$ requires at least one erroneous spacetime qubit at that coordinate, and
\begin{equation}
|\rho_{\rm S}(E)|\leq |E|.
\label{eq:steanecontraction}
\end{equation}
If $E$ belongs to the dressed class $c$, gauge invariance gives $\rho_{\rm S}(E)\in c$ and therefore $|E| \geq d_c$. Conversely, a minimum-weight $v_c\in c$ placed as $P_{d_o}^{(2)}(v_c)$ has weight $d_c$ and represents the same dressed class. Thus the class minimum is exactly $d_c$. Minimizing over $c\neq W_{\rm CSS}$ gives $d_{\mathrm{st}}=d$. 

\end{proof}

\subsection{Knill error-correction gadget}
\label{sec:knill}

The Knill gadget combines transversal physical Bell measurements with an encoded Bell state ~\cite{Knill2004}. The qubit labeling convention is the same as the previous section.

Let the underlying code be a general $[\![n,k,d]\!]$ stabilizer code with
\begin{equation}
\St_{\rm code}=\left\langle P(s_l)
:\quad \forall l=1,\ldots,r \right\rangle,
\label{eq:knillcodegens}
\end{equation}
where $r=n-k$, $s_l\in\F^{2n}$ are independent symplectic Pauli labels, and
\begin{equation}
W=\Span\{s_l\}_{l=1}^{r}.
\label{eq:knillW}
\end{equation}
Choose a symplectic logical basis
\begin{equation}
W^\perp=W\oplus\Span\{\ell_{X,\mu},\ell_{Z,\mu}\}_{\mu=1}^{k},
\qquad \lambda(\ell_{X,\mu},\ell_{Z,\nu})=\delta_{\mu\nu}.
\label{eq:knillbasis}
\end{equation}

At $t=0$, the data block contains the state to be corrected and blocks $a$ and $b$ are prepared in the encoded Bell state
\begin{equation}
\ket{\overline\Phi_k}_{ab}=2^{-k/2}\sum_{x\in\F^k}
\ket{\overline x}_a\ket{\overline x}_b.
\label{eq:bellstate}
\end{equation}
At $t=1$, a transversal physical $\mathrm{CNOT}_{d\to a}$ is applied. The data output legs are measured in $X$ basis, the output legs of codeblock $a$ are measured in $Z$ basis, and the codeblock $b$ survives as the output block $b_o^{(1)}$.

\begin{figure}[t]
    \centering
    \includegraphics[width=0.45\linewidth]{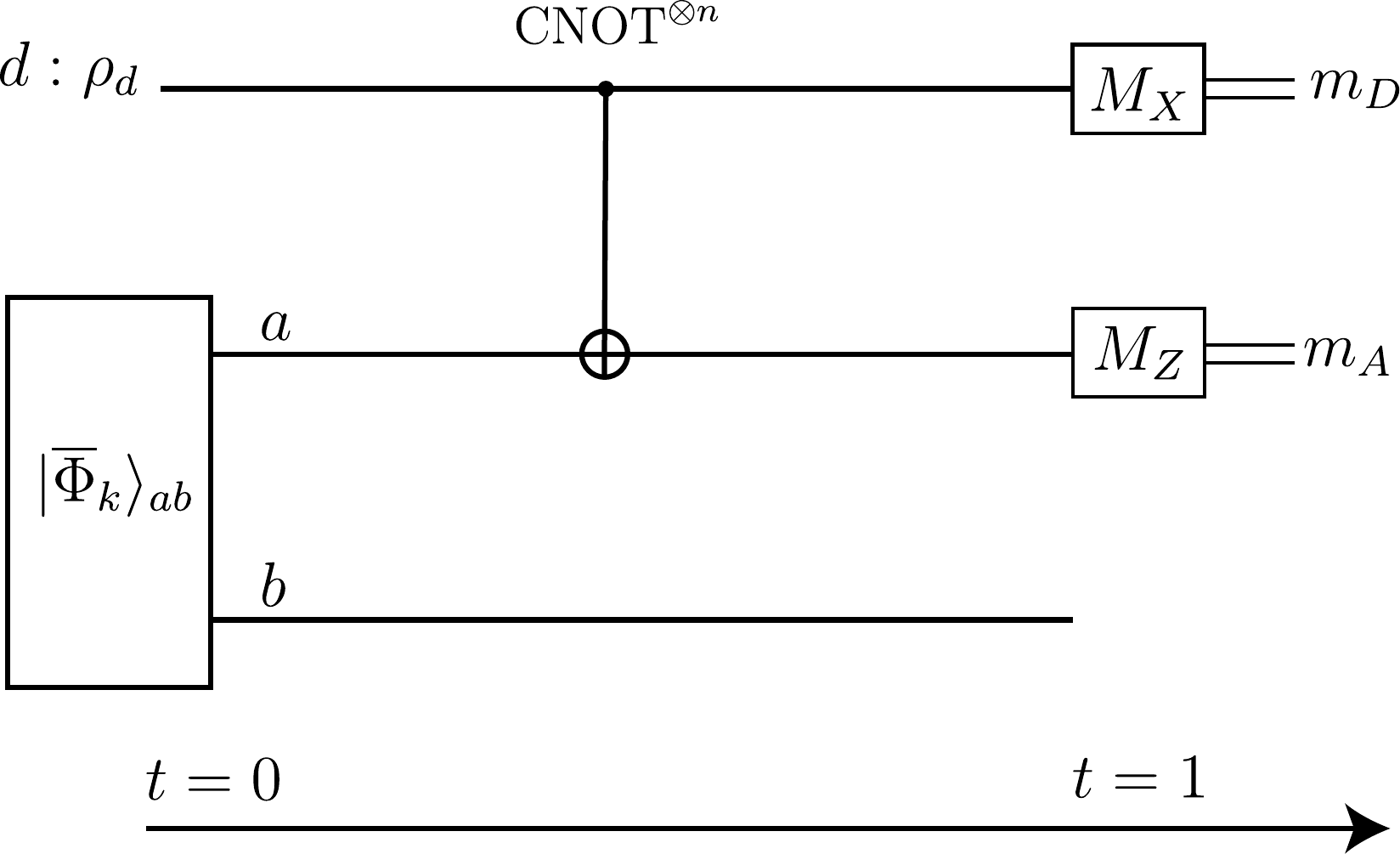}
    \caption{General Knill extraction gadget in the compact timing convention. The data and $a$ blocks are destructively measured directly on the output legs of the transversal CNOT layer; block $b$ survives as the quantum output.}
\label{fig:knill}
\end{figure}

The spacetime code representation of the Knill EC gadget contains
\begin{equation}
\begin{array}{c|l}
0&d_o^{(0)},\ a_o^{(0)},\ b_o^{(0)},\\
1&d_i^{(1)},\ a_i^{(1)},\ b_i^{(1)},\ d_o^{(1)},\ a_o^{(1)},\ b_o^{(1)}.
\end{array}
\label{eq:knillblocks}
\end{equation}
Hence $N=9n$.

The gauge group of the associated spacetime subsystem code, $\mathcal{G}=\langle \mathcal{S}_{\mathrm{in}},\mathcal{B},\mathcal{R},\mathcal{M}\rangle$, are generated by Equations~\eqref{eq:knillprep}--\eqref{eq:knillmeasure} . The preparation gauge set $\mathcal{S}_{\mathrm{in}}$ includes

\begin{align}
& P_{d_o}^{(0)}(s_l),\qquad P_{a_o}^{(0)}(s_l),\qquad P_{b_o}^{(0)}(s_l),
&&l=1,\ldots,r,\notag\\
&P_{a_o}^{(0)}(\ell_{X,\mu})P_{b_o}^{(0)}(\ell_{X,\mu}),
\qquad P_{a_o}^{(0)}(\ell_{Z,\mu})P_{b_o}^{(0)}(\ell_{Z,\mu}),
&&\mu=1,\ldots,k.
\label{eq:knillprep}
\end{align}
The first group consists of the input code stabilizers, and the remaining groups specify the logical Bell state of the ancilla blocks.

The bond gauge group is
\begin{align}
\mathcal{B}=\langle &X_{d_o}^{(0)}(e_q)X_{d_i}^{(1)}(e_q),~
 Z_{d_o}^{(0)}(e_q)Z_{d_i}^{(1)}(e_q),~ X_{a_o}^{(0)}(e_q)X_{a_i}^{(1)}(e_q),~
 Z_{a_o}^{(0)}(e_q)Z_{a_i}^{(1)}(e_q),\\
 &X_{b_o}^{(0)}(e_q)X_{b_i}^{(1)}(e_q),~
 Z_{b_o}^{(0)}(e_q)Z_{b_i}^{(1)}(e_q):~\forall q=1,\ldots,n\rangle.
\label{eq:knillwires}
\end{align}

The gate gauges are contributed by the transversal gate $\mathrm{CNOT}_{d\to a}$  and the idling gate on block $b$. 
\begin{align}
\mathcal{R}=\langle
&X_{d_i}^{(1)}(e_q)X_{d_o}^{(1)}(e_q)X_{a_o}^{(1)}(e_q),&
&Z_{d_i}^{(1)}(e_q)Z_{d_o}^{(1)}(e_q),
&&X_{a_i}^{(1)}(e_q)X_{a_o}^{(1)}(e_q),&
&Z_{a_i}^{(1)}(e_q)Z_{d_o}^{(1)}(e_q)Z_{a_o}^{(1)}(e_q),\\
& X_{b_i}^{(1)}(e_q)X_{b_o}^{(1)}(e_q),&&
Z_{b_i}^{(1)}(e_q)Z_{b_o}^{(1)}(e_q):&&\forall q=1,\ldots,n\rangle.
\label{eq:knillgate}
\end{align}

The final measurements are imposed directly on the gate output legs:
\begin{equation}
\mathcal{M}=\langle X_{d_o}^{(1)}(e_q),\quad Z_{a_o}^{(1)}(e_q): \quad \forall q=1,\ldots,n\rangle.
\label{eq:knillmeasure}
\end{equation}

A complete independent stabilizer generating set can be obtained by solving the centralizer equation, 
\begin{equation}
\St
=\left\langle S_{\de,l},S_{B,l}: \forall l=1,\ldots,r\right\rangle,
\label{eq:knillstabilizers}
\end{equation}
with 
\begin{align}
S_{\de,l}=&(-1)^{x_l\cdot z_l}P_{d_o}^{(0)}(s_l)P_{a_o}^{(0)}(s_l)
P_{d_i}^{(1)}(s_l)P_{a_i}^{(1)}(s_l)
X_{d_o}^{(1)}(x_l)Z_{a_o}^{(1)}(z_l),\notag\\
S_{B,l}={}&P_{b_o}^{(0)}(s_l)P_{b_i}^{(1)}(s_l)P_{b_o}^{(1)}(s_l).
\label{eq:knillexplicitstabs}
\end{align}
Similar to the Steane gadget case, the stabilizers specifying the input state are redundant, and the algebra below also applies when the input is unrestricted. If the data input is an encoded logical state in the same signed code sector as the Bell resource, the corresponding fault-free measurement parity is $+1$. For an unrestricted input, the same relation measures an unknown input syndrome rather than defining an unconditional detector.

In summary, the spacetime subsystem code parameters are
\begin{equation}
N=9n,\quad \dim\G =16n,\quad
\dim\St=2(n-k),\quad K=k,\quad R=7n+k.
\label{eq:knillranks}
\end{equation}

\begin{proposition}[fault distance of the Knill EC gadget]
The spacetime subsystem code associated with the Knill error correction gadget is a subsystem code with 
\begin{equation}
[\![N,K,R,d_{\st}]\!]=[\![9n,k,7n+k,d]\!].
\label{eq:knillparameters}
\end{equation}
\end{proposition}

\begin{proof}
Let $e_{d_o}^{(0)}$, $e_{a_o}^{(0)}$, etc., denote the $2n$-bit symplectic Pauli label of an error $E$ on the indicated block. Define
\begin{equation}
\pi_X(x\mid z)=(x\mid0),\qquad
\pi_Z(x\mid z)=(0\mid z),
\label{eq:projections}
\end{equation}
and
\begin{align}
a(E)={}&e_{d_o}^{(0)}+e_{a_o}^{(0)}+e_{d_i}^{(1)}+e_{a_i}^{(1)}
+\pi_Z(e_{d_o}^{(1)})+\pi_X(e_{a_o}^{(1)}),\notag\\
b(E)={}&e_{b_o}^{(0)}+e_{b_i}^{(1)}+e_{b_o}^{(1)},\notag\\
\rho_{\rm K}(E)={}&a(E)+b(E).
\label{eq:knillrho}
\end{align}
The complete stabilizer syndrome is
\begin{equation}
\lambda(E,S_{\de,l})=\lambda(a(E),s_l),\qquad
\lambda(E,S_{B,l})=\lambda(b(E),s_l).
\label{eq:knillsyndrome}
\end{equation}
Thus $E\in\St^\perp$ implies $a(E),b(E)\in W^\perp$ and therefore $\rho_{\rm K}(E)\in W^\perp$. Moreover, according to Corollary \ref{coro:singlemeasreduc}, 
\begin{equation}
E\sim_{\G} P_{b_o}^{(1)}(\rho_{\rm K}(E)),\qquad
\St^\perp/\G\simeq W^\perp/W.
\label{eq:knillquotient}
\end{equation}

Every physical coordinate of $\rho_{\rm K}(E)$ depends only on errors at the same coordinate in the three physical blocks and their spacetime copies. The projections in Eq.~\eqref{eq:projections} can remove support but cannot create support at a different coordinate, so
\begin{equation}
|\rho_{\rm K}(E)| \leq |E|.
\label{eq:knillcontraction}
\end{equation}
Fix a nontrivial logical coset $c=v+W$ and let $d_c=\min_{w\in W} |v+w|$. Every zero-syndrome spacetime representative of this dressed class maps into $c$, so Eq.~\eqref{eq:knillcontraction} gives $|E| \geq d_c$. Conversely, a minimum-weight representative $v_c\in c$ placed as $P_{b_o}^{(1)}(v_c)$ is a zero-syndrome representative of weight $d_c$. Hence the class minimum is exactly $d_c$, and minimizing over nontrivial cosets gives $d_{\st}=d$. 
\end{proof}

\subsection{Single-shot error correction}

Single-shot error correction \cite{bombin2015single-shot,campbell2019theory,kubica2022single} provides an alternative to the conventional strategy of detecting measurement errors through repeated syndrome extraction. By exploiting suitable redundancy among stabilizer or gauge checks, a single-shot error-correction protocol can perform recovery using only a single round of noisy check measurements. The linear relations among the measurement outcomes, called metachecks, allow measurement faults to be detected without repeating the entire syndrome extraction circuit. In this section, we consider the single-shot syndrome extraction of a stabilizer code with $(\gamma,f)$-soundness \cite{campbell2019theory} with a two-step recovery where we first use metachecks to correct measurement faults and then use the corrected stabilizer syndrome to perform recovery, but the analysis can be directly generalized to single-shot subsystem codes.

\begin{definition}[$(\gamma,f)$-soundness \cite{campbell2019theory}]
Let $\sigma$ denote the syndrome map associated with a set of
measured stabilizer checks. The check set is said to be
$(\gamma,f)$-sound if, for every Pauli error $E$ whose syndrome satisfies
$|\sigma(E)|=x \leq  \gamma$, there exists a Pauli error $E^\star$ with the
same syndrome such that
\begin{equation}
    \sigma(E^\star)=\sigma(E),
    \qquad
    |E^\star|\leq f(x).
\end{equation}
In other words, every sufficiently low-weight valid syndrome admits
a Pauli representative whose weight is upper bounded by a monotonically increasing soundness
function $f$.
\end{definition}

Let $H_X$ and $H_Z$ denote the measured $X$- and $Z$-check matrices with code distance being $d$. The corresponding metacheck matrices $M_X$ and $M_Z$ satisfy
\begin{eqs}
    M_X H_X=0,\quad M_Z H_Z=0.
\end{eqs}
For example, the ideal \(X\)-check syndrome $s_X=H_Xe_Z$ and $Z$-check syndrome $s_Z=H_Z e_X$ satisfy
\begin{eqs}
    M_X s_X=0,\quad M_Z s_Z=0.
\end{eqs}
If the observed $X$- and $Z$-syndromes is affected by a
measurement fault, the syndrome becomes
\begin{eqs}
    s_X'= s_X+ m_X, \quad s_Z'= s_Z + m_Z.
\end{eqs}
The metasyndromes are therefore
\begin{eqs}
    M_X s_X'= M_X m_X, \quad M_Z s_Z'= M_Z m_Z.
\end{eqs}
Thus, $M_X s_X'$ detects inconsistencies caused by faults in the \(X\)-check measurement outcomes. The same argument applies to \(M_Z\) and the \(Z\)-check measurement outcomes. In this way, single-shot error correction replaces temporal redundancy across repeated rounds with algebraic redundancy among checks measured within a single round.

Under the phenomenological error model, the single-shot protocol is represented by the syndrome extraction gadget shown in Fig. \(\ref{fig:SE_gadget}\), but it contains only one \(X\)-syndrome extraction gadget and one \(Z\)-syndrome extraction gadget. Because there are no consecutive measurement rounds whose outcomes can be compared, the bulk detectors are given by the metachecks \(M_X\) and \(M_Z\). The bulk detector matrix therefore has the block-diagonal form

\begin{eqs}
    D= \begin{pmatrix}
        M_X & 0\\
        0& M_Z
    \end{pmatrix}.
\end{eqs}
Its detecting region consists of the measurement fault coordinates appearing in these two metacheck matrices.

Now consider $X$- and $Z$-measurement faults \(\eta_X\) and $\eta_Z$. Let \(\widetilde{\eta_X}\) and $\widetilde{\eta_Z}$ be the minimum-weight estimate of the measurement faults $\eta_X, \eta_Z$ satisfying
\begin{eqs}
    M_X \widetilde{\eta_X}= M_X \eta_X,\quad M_Z \widetilde{\eta_Z} =M_Z \eta_Z.
\end{eqs}
After we correct the measurement fault, the residual measurement fault in the $X$ and $Z$ sectors become $\eta_X+\widetilde{\eta_X}$ and $\eta_Z+ \widetilde{\eta_Z}$. They satisfy 
\begin{eqs}
    M_X (\eta_X+ \widetilde{\eta_X})=0, \quad M_Z (\eta_Z+ \widetilde{\eta_Z})=0.
\end{eqs}
Because $\widetilde{\eta_X}$, $\widetilde{\eta_Z}$ are minimum-weight recoveries of $M_X \eta_X$ and $M_Z \eta_Z$ respectively, they satisfy $|\widetilde{\eta_X}| \leq |\eta_X|$ and $|\widetilde{\eta_Z}| \leq |\eta_Z|$.
The residual measurement fault satisfies
\begin{eqs}
    |\eta_X + \widetilde{\eta_X}| \leq 2|\eta_X|,\quad |\eta_Z + \widetilde{\eta_Z}| \leq 2|\eta_Z|.
\end{eqs}

Therefore for measurement faults $\eta_X,\eta_Z$ satisfying $2|\eta_X| \leq \gamma$ and $2|\eta_Z| \leq  \gamma$, the residual measurement faults $\eta_X+ \widetilde{\eta_X}$ and $\eta_Z+ \widetilde{\eta_Z}$ are valid stabilizer syndromes. The soundness gives the minimum-weight recovery $\varepsilon_Z$ and $\varepsilon_X$ of Pauli syndrome $\eta_X+ \widetilde{\eta_X}$ and $\eta_Z+ \widetilde{\eta_Z}$ respectively, they satisfy
\begin{eqs}
    &|\varepsilon_Z| \leq f(|\eta_X+ \widetilde{\eta_X}|) \leq f(2|\eta_X|),\\
    &|\varepsilon_X|  \leq f(|\eta_Z+ \widetilde{\eta_Z}|) \leq f(2|\eta_Z|).
\end{eqs}

Define the generalized inverse of the soundness function \(f\) by
\begin{eqs}
    f_{\leq }^{-1}(t):= \max \left \{x \in \mathbb{N}: f(x) \leq t \right\}.
\end{eqs}
The largest residual measurement-error weight for which the corresponding minimum-weight Pauli recovery of repaired syndromes $\eta_X+\widetilde{\eta_X}$ and $\eta_X+\widetilde{\eta_Z}$ are guaranteed to be correctable is
\begin{eqs}
    x^\star=\min \left\{\gamma, f_{\leq}^{-1}(t)\right\}.
\end{eqs}

Since the residual measurement fault can have weight as large as twice the original measurement fault weight, the bulk effective distance is 
\begin{eqs}
    d_{\mathrm{bulk}}= \left\lfloor \frac{x^\star}{2}  \right\rfloor .
\end{eqs}

In the presence of both input data errors \(\varepsilon_X, \varepsilon_Z\) and measurement errors $m_X, m_Z$, we may use the two-step recovery in the proof of Theorem~\ref{thm:ECCP} and obtain the sufficient conditions of satisfying ECCP are
\begin{eqs}\label{eq:single_shot_correctness_condition}
\begin{cases}
    |\varepsilon_Z| + f(2|\eta_X|) \leq t,\\
    |\varepsilon_X| + f(2|\eta_Z|) \leq t,\\
    2|\eta_X|\leq \gamma,\\
    2|\eta_Z|\leq \gamma,
\end{cases}
\end{eqs}
because the output residual fault is written as $\mathbf{0} \oplus \mathbf{0} \oplus \varepsilon_{\mathrm{out}}(\varepsilon_Z, \eta_X)$ and $\mathbf{0} \oplus \mathbf{0} \oplus \varepsilon_{\mathrm{out}}(\varepsilon_X, \eta_Z)$ with $|\varepsilon_{\mathrm{out}}(\varepsilon_Z, \eta_X)|\leq |\varepsilon_Z|+f(2|\eta_X|)$ and $|\varepsilon_{\mathrm{out}}(\varepsilon_X, \eta_Z)| \leq |\varepsilon_X|+f(2|\eta_Z|)$ in the $X$ and $Z$ sectors. Under the conditions in Eq.~\eqref{eq:single_shot_correctness_condition}, these output residual faults can be detected by the stabilizer tubes and corrected by minimum-weight recovery on the output boundary. The gadget therefore satisfies ECCP.


Below we provide a theorem that relates soundness to fault amplification, allowing one to use Corollary~\ref{cor:amplified_ECCP} to establish the correctness property of a single-shot error correction gadget.

\begin{theorem}[Soundness implies bounded fault amplification]\label{thm:soundness_amplification.}
    Let $\sigma(\varepsilon_X\oplus \varepsilon_Z)
:=
H_Z \varepsilon_X\oplus H_X \varepsilon_Z$
denote the syndrome map of a single-shot CSS code, and suppose that its
measured check set is $(\gamma,f)$-sound. Consider a single-shot
syndrome extraction gadget under the phenomenological noise model.
Assume that the syndrome-repair  procedure maps every measurement fault
configuration $\eta=\eta_Z \oplus \eta_X$, where $\eta_X, \eta_Z$ describe faults in the $X$- and $Z$-check measurements, to a valid stabilizer syndrome $r(\eta)= \eta+\widetilde{\eta}\in\operatorname{im}\sigma $ by applying a recovery $\widetilde{\eta}$ of measurement fault,
satisfying $|r(\eta)|\leq c|\eta| $
for some constant $c\geq 1$. Then the single-shot error correction gadget after syndrome repair and minimum-weight recovery is effectively 
$(\gamma/c,F)$-bounded fault amplifying, with amplification function
\begin{equation}
F(x):=2x+f(cx).
\end{equation}
\end{theorem}

\begin{proof}
Under the phenomenological noise model, the input, bulk, and output
regions consist of the input data ports, measurement-fault locations,
and output data ports, respectively. Let
\begin{equation}
\varepsilon= \varepsilon_X\oplus \varepsilon_Z
\end{equation}
be an input Pauli fault whose $X$ and $Z$ components are $\varepsilon_X$ and $\varepsilon_Z$ respectively. 

Propagating an input fault $\varepsilon$ through an
syndrome extraction gadget gives the gauge equivalence
\begin{equation}
\varepsilon\oplus\mathbf{0}\oplus\mathbf{0}
\sim_{\mathcal G}
\mathbf{0}\oplus \sigma(\varepsilon)\oplus \varepsilon,
\end{equation}
where $\sigma(\varepsilon_X \oplus \varepsilon_Z)$ is the measurement fault configuration supported on the checks
violated by $\varepsilon$.

Now consider a pure measurement fault $\eta=\eta_Z\oplus \eta_X\in A_{1,\mathrm{bulk}}$. After syndrome repairing, its remaining effect is
described by the repaired syndrome $r(\eta)=\eta +\widetilde{\eta}$ where $\widetilde{\eta}$ is a recovery of measurement fault and we estimate $\eta_X, \eta_Z$ separately. Because the
check set is $(\gamma,f)$-sound, whenever $|r(\eta)| \leq \gamma$, the measurement fault $\mathbf{0} \oplus (\eta + \widetilde{\eta}) \oplus \mathbf{0}$ is gauge equivalent to  $\varepsilon_\eta \oplus \mathbf{0} \oplus \varepsilon_\eta$ where the weight of $\varepsilon_\eta$ satisfies
\begin{equation}
\sigma(\varepsilon_\eta)=r(\eta),
\qquad
|\varepsilon_\eta|\leq f\left(|r(\eta)|\right) \leq f(c|\eta|).
\end{equation}

Hence, the repaired syndrome $r(\eta)$ has a Pauli realization
$\varepsilon_\eta$ of weight at most $f(c|\eta|)$. Combining this
representative with the propagation of the input fault $\varepsilon$, we obtain
a gauge-equivalent normal form
\begin{equation}
\varepsilon\oplus(\eta + \widetilde{\eta})\oplus\mathbf{0}
\sim_{\mathcal G}
\mathbf{0}\oplus (\sigma(\varepsilon)+\eta+\widetilde{\eta})
\oplus \varepsilon,
\end{equation}
where $\sigma(\varepsilon)+\eta+\widetilde{\eta}$ is supported entirely on the bulk measurement
locations. We then apply a minimum-weight recovery $\varepsilon'$ to the output data qubits, chosen independently in the $X$ and $Z$ sectors and satisfying$\sigma(\varepsilon')= \sigma(\varepsilon)+ \eta + \widetilde{\eta}$. The soundness condition gives $|\varepsilon_\eta | \leq f(c|\eta|)$, hence the minimum-weight recovery on data qubits satisfies
\begin{eqs}
    |\varepsilon' | \leq |\varepsilon+ \varepsilon_\eta| \leq |\varepsilon | + f(c| \eta|).
\end{eqs}
The final fault configuration becomes $\mathbf{0} \oplus m(\sigma,\eta) \oplus \varepsilon_{\mathrm{out}} \equiv \mathbf{0} \oplus (\sigma(\varepsilon)+ \eta+ \widetilde{\eta}) \oplus (\varepsilon+\varepsilon')$ where $|\varepsilon_{\mathrm{out}}| \leq |\varepsilon+\varepsilon'| \leq |\varepsilon|+|\varepsilon'| \leq 2|\varepsilon|+ f(c|\eta|)$.

Setting
\begin{equation}
x:=|\varepsilon_X|+|\eta_Z|,\quad 
z:=|\varepsilon_Z|+|\eta_X|
\end{equation}
and using the monotonicity of $f$, we find the $X$- and $Z$-component of $\varepsilon_{\mathrm{out}}(\varepsilon,\eta)$ are written as
\begin{eqs}
|\varepsilon_{\mathrm{out},X}|
&\leq 2|\varepsilon_X|+f(c|\eta_Z|) \leq 2x+ f(cx)
=F(x),\\
|\varepsilon_{\mathrm{out},Z}|
&\leq 2|\varepsilon_Z|+f(c|\eta_X|) \leq 2z + f(cz)= F(z).
\end{eqs}
The conditions $x,z<\gamma/c$ ensure that soundness is applicable. Therefore, the gadget is
$(\gamma/c,F)$-bounded fault amplifying.
\end{proof}

For a minimum-weight syndrome repair, we have
$|\widetilde{\eta}|\leq|\eta|$, and consequently $|r(\eta)|
=
|\eta + \widetilde{\eta}|
\leq 2|\eta|$ by definition.
Thus, a minimum-weight syndrome repairing gives $c=2$. It follows from preceding theorem that, under the phenomenological noise model, the single-shot syndrome extraction gadget is $(\gamma/2,F)$-bounded fault amplifying, with amplification function
\begin{eqs}
F(x)=2x+f(2x),
\end{eqs}
acting on the $X$ and $Z$ sectors separately.

\section{Algebraic Correspondence between Circuit-Based and Measurement-Based Fault Tolerance}

It has been widely believed that CSS Clifford circuits, fault complexes associated with bipartite graph states, and CSS (phase-free) ZX diagrams are closely related descriptions of the same underlying fault tolerance structure in spacetime. For instance, repeated syndrome extraction circuit represented in the ZX-calculus \cite{bombin2024unifying} and the corresponding foliated cluster states \cite{foliation} exhibit similar geometric structures and are generally regarded as equivalent. However, to our knowledge, a rigorous derivation establishing the unification of all three objects is lacking.

In this section, we make this connection explicit by deriving a cluster state measurement-based realization from a circuit-based CSS protocol represented as a ZX diagram. For protocols whose spacetime complexes admit a direct-sum decomposition into $X$ and $Z$ sectors and whose tensor networks are bipartite, we construct a chain map to the mathematical formulation of measurement-based fault tolerance, known as a fault complex \cite{foliation,faultcomplex}. Algebraically, this construction relates the two descriptions through the $X/Z$ factorization of the spacetime complex. Graphically, the cluster state inherits the connectivity of the circuit's $X$ and $Z$ tensors. We first review these tensors, then present the chain complex construction, and finally analyze the map microscopically to establish the correspondence between circuit $X/Z$ detectors and measurement-based primal/dual detectors.


\subsection{\texorpdfstring{$X$}{} and \texorpdfstring{$Z$}{} tensors}

A CSS Clifford circuit can be represented as a tensor network constructed by contracting \(X\) and \(Z\) tensors. An $n$-leg $Z$ tensor has gauge group
\begin{eqs}\label{eq:Rz}
    \mathcal{R}_Z =\langle \bigotimes_{i=1}^n X_i, Z_1 Z_2,\ldots, Z_{n-1} Z_n \rangle.
\end{eqs}
Thus, a \(Z\) tensor has one \(X\)-type gauge generator and a $n-1$ weight-two \(Z\)-type gauge generators.

Similarly, an $n$-leg $X$ tensor has gauge group
\begin{eqs}\label{eq:Rx}
    \mathcal{R}_X=\langle \bigotimes_{i=1}^n Z_i, X_1 X_2, \ldots, X_{n-1} X_n \rangle.
\end{eqs}
An \(X\) tensor has one \(Z\)-type gauge generator and $n-1$ weight-two \(X\)-type gauge generators.

The $X$ and $Z$ tensors are also called $X$ and $Z$ spiders in the language of ZX calculus \cite{bombin2024unifying,rodatz2025fault}. We use $\mathcal{R}$ to denote the collection of local gauge groups of individual tensors
\begin{eqs}
    \mathcal{R}=\langle  \mathcal{R}_1,\ldots, \mathcal{R}_L \rangle,
\end{eqs}
where $L$ is the number of individual tensors in the circuit and $\mathcal{R}_j$ is the gauge group of the $j$-th individual tensor. 

Throughout this section, we use the following notation to represents different types of gauge generator of tensors:
\begin{itemize}
    \item $g_{XX}$: \(X\)-type gauge generators of \(X\) tensors;
    \item $g_{ZZ}$: \(Z\)-type gauge generators of \(Z\) tensors;
    \item $g_{XZ}$: \(Z\)-type gauge generators of \(X\) tensors;
    \item $g_{ZX}$: \(X\)-type gauge generators of \(Z\) tensors;
\end{itemize}
In this notation, the first index specifies the tensor type and the second specifies the Pauli type.

\begin{figure}[h]
    \centering
    \includegraphics[width=0.28\textwidth]{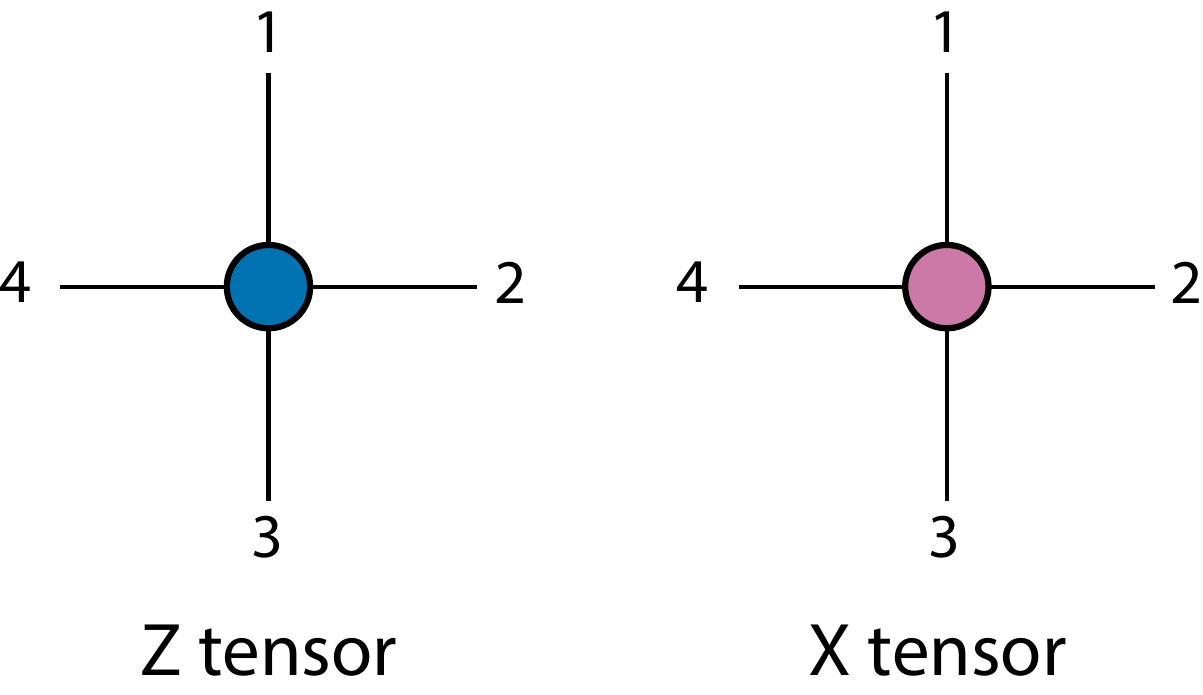}
    \caption{Illustration of the 4-leg $X$ and $Z$ tensors.}
    \label{fig:XZ_tensor}
\end{figure}

\subsection{Fault complexes from spacetime complexes}

\begin{theorem}
    For a CSS quantum circuit whose spacetime complex has a direct-sum structure, the circuit can be mapped to a bipartite graph state described by fault complex $\mathcal{F}$, namely a 4-term chain complex that describes the fault-tolerance properties of the protocol.
\end{theorem}

We consider a CSS circuit associated with a CSS code whose spacetime tensor network representation consists of individual $X$  and $Z$ tensors. We assume that this representation is bipartite: every internal edge connects an $X$ tensor to a $Z$ tensor, so tensors of the same type are never directly connected.

Before constructing the spacetime complex, we gauge fix and eliminate the quantum degrees of freedom at the input, output, and measurement ports to obtain a reduced spacetime tensor. Each such port of an $X$ tensor is fixed to the $Z=+1$ eigenstate, while each such port of a $Z$ tensor is fixed to the $X=+1$ eigenstate. Gauge fixing a port reduces an $n$-leg tensor to an $(n-1)$-leg tensor of the same type. The gauge generators, fault positions, and detectors of the resulting reduced spacetime tensor define the chain complex.

\begin{equation}
\begin{tikzcd}[row sep=1.2em, column sep=1.0em]
\mathcal{A}_\bullet:&
A_2 \arrow[r, "\partial_{2,A}"] &
A_1 \arrow[r, "\partial_{1,A}"] &
A_0, \\
{} &
{\scriptstyle \text{gauge checks}} &
{\scriptstyle \text{fault positions}} &
{\scriptstyle \text{detector syndromes}}
\end{tikzcd}
\end{equation}
where $\partial_2=\begin{pmatrix}
    G_X^\mathsf{T} & 0\\
    0& G_Z^\mathsf{T}
\end{pmatrix}$, $\partial_1=\begin{pmatrix}
    S_Z & 0\\
    0 & S_X
\end{pmatrix}$. 
The chain complex condition $\partial_{1,A} \partial_{2,A}=0$ requires that $S_Z G_X^\mathsf{T}=0$ and $S_X G_Z^\mathsf{T}=0$ implying that the $X$- and $Z$-detectors always commute with $Z$- and $X$-gauge checks.

We can see this chain complex $\mathcal{A}_\bullet$ can be decomposed into direct sum of two complexes $\mathcal{A}_\bullet =\mathcal{C}_\bullet  \oplus \mathcal{E}_\bullet $
\begin{equation}
\label{eq:factorized_spacetime_complex}
\begin{aligned}
 &\begin{tikzcd}[row sep=1.2em, column sep=1.0em]\mathcal{C}_\bullet:&
C_{2} \arrow[r, "G_X^\mathsf{T}"] &
C_{1} \arrow[r, "S_Z"] &
C_{0}, \\
{} &
{\scriptstyle \text{$X$ gauge checks}} &
{\scriptstyle \text{$X$ fault positions}} &
{\scriptstyle \text{$Z$ detector syndromes}}
\end{tikzcd}\\
&\begin{tikzcd}[row sep=1.2em, column sep=1.0em]\mathcal{E}_\bullet:&
E_{2} \arrow[r, "G_Z^\mathsf{T}"] &
E_{1} \arrow[r, "S_X"] &
E_{0}. \\
{} &
{\scriptstyle \text{$Z$ gauge checks}} &
{\scriptstyle \text{$Z$ fault positions}} &
{\scriptstyle \text{$X$ detector syndromes}}
\end{tikzcd}
\end{aligned}
\end{equation}
The $X$- and $Z$-detector matrices $D_X$ and $D_Z$ consist of selected rows of the $X$- and $Z$-spacetime stabilizer matrices $S_X$ and $S_Z$, respectively. As we show below, the chain map identifies $D_X$ with the primal detectors $D_{\mathrm{primal}}$ and $D_Z$ with the dual detectors $D_{\mathrm{dual}}$ of the fault complex. 

We can map the spacetime complex to a `twisted' fault complex by following chain maps $f_2,f_1,f_0$
\begin{equation}
\label{eq:ST_complex_chain_map}
\begin{tikzcd}[row sep=2.0em, column sep=2.0em]
A_2 \arrow[r, "\partial_{2,A}"] \arrow[d, "f_2"] &
A_1 \arrow[r, "\partial_{1,A}"] \arrow[d, "f_1"]&
A_0 \arrow[d, "f_0"] \\
F^1 \oplus F_2 \arrow[r, "\partial_{2,\mathbf{F}}"] &
F^2 \oplus F_1 \arrow[r, "\partial_{1,\mathbf{F}}"] &
F^3 \oplus F_0,
\end{tikzcd}
\end{equation}
where the boundary maps of the twisted fault complex are $\partial_{2,\mathbf{F}} =\begin{pmatrix}
     \partial_{2,F}^\mathsf{T} & 0\\
    0 & \partial_{2,F}
\end{pmatrix}$, $\partial_{1,\mathbf{F}} =\begin{pmatrix}
    \partial_{3,F}^\mathsf{T} & 0\\
    0 & \partial_{1,F}
\end{pmatrix}$. And the boundary map conditions $\partial_{1,\mathbf{F}}\partial_{2,\mathbf{F}}=0$ is equivalent to the boundary map conditions of conventional fault complex $\partial_{1,F}\partial_{2,F}=0, \partial_{3,F}^\mathsf{T} \partial_{2,F}^\mathsf{T}=0$. Here $\partial_{2,F}$ is the connectivity matrix between dual and primal fault positions, and $\partial_{1,F}, \partial_{3,F}^\mathsf{T}$ are primal and dual detector matrices. The detailed form of primal and dual detector matrices will be provided in the next subsection.

The chain map $f_2, f_1, f_0$ can be summarized as
\begin{itemize}
    \item Action of $f_2$:  The map $f_2$ acts on gauge checks. It sends the following generators to zero
    \begin{itemize}
        \item the weight-two \(X\)-type generators \(g_{XX}\) of the \(X\) tensors;
        \item the weight-two \(Z\)-type generators \(g_{ZZ}\) of the \(Z\) tensors;
        \item all generators belonging to the bond gauge group \(\mathcal B\).
    \end{itemize}
    The remaining generators are mapped nontrivially:
    \begin{itemize}
    \item the \(X\)-type generators of the \(Z\) tensors \(g_{ZX}\) , are mapped to the primal positions in \(F^1\);
    \item the \(Z\)-type generators of the \(X\) tensors \(g_{XZ}\), are mapped to the dual positions in \(F_2\).
    \end{itemize}
    Equivalently, each \(Z\) tensor is associated with a primal qubit, and each \(X\) tensor is associated with a dual qubit.

    \item Action of $f_1$: The map \(f_1\) acts on individual fault positions. A pictorial description of $f_1$ and how it acts on the gauge checks is provided in Fig.~\ref{fig:chain_map}.
    
    An \(X\) fault on an edge is mapped to a \(Z\) fault on the dual qubit associated with the incident \(X\) tensor. A \(Z\) fault is mapped to a \(Z\) fault on the primal qubit associated with the incident \(Z\) tensor. For an edge shared by an \(X\) tensor and a \(Z\) tensor, the faults on its two half-edges satisfy 
    \begin{eqs}
        f_1(XI)=f_1(IX)= Z_{\mathrm{dual}}, \quad f_1(ZI)=f_1(IZ)= Z_{\mathrm{primal}}.
    \end{eqs}
    Consequently, every weight-two gauge generators, including $g_{XX}, g_{ZZ}$ and elements in $\mathcal{B}$, are mapped to the identity:
    \begin{eqs}
        f_1(XX)=Z_{\mathrm{dual}}^2= \mathbb{I}, \quad  f_1(ZZ)=Z_{\mathrm{primal}}^2= \mathbb{I}.
    \end{eqs}
    Therefore, only two types of local gauge generators survive: $g_{XZ}$ and $g_{ZX}$. These surviving generators become the \(Z\)-components of the cluster state stabilizers associated with the primal and dual qubits.

    \item $f_0$: identity map where the primal and dual syndromes corresponds to the $X/Z$ detector syndromes.
\end{itemize}

\begin{figure}[t]
    \centering
    \includegraphics[width=1.0\textwidth]{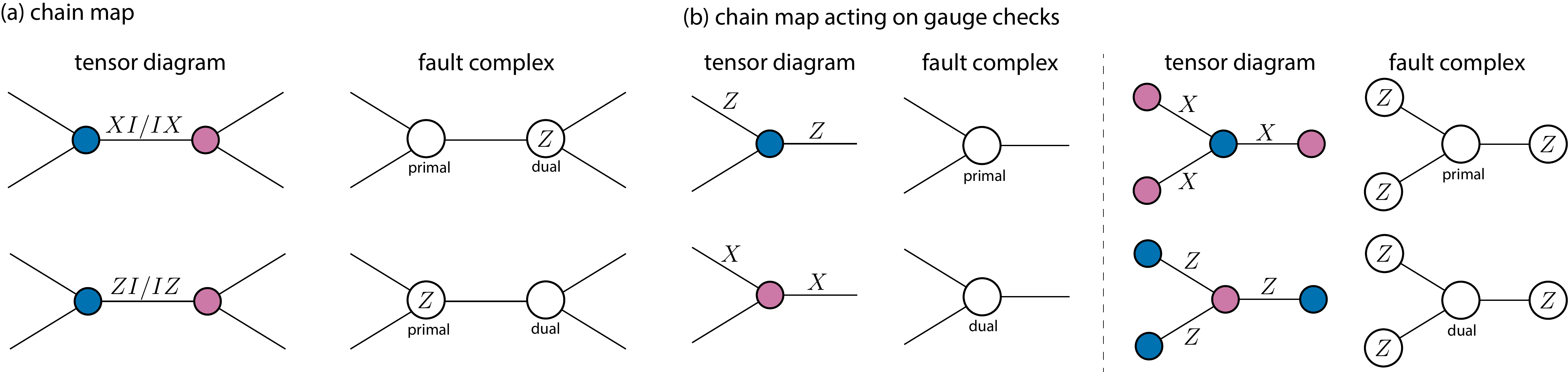}
    \caption{Pictorial representation of the chain map $f_1$ that maps fault positions in spacetime complex and tensor (ZX) diagram to fault positions fault complex. The chain map $f_1$ maps a $Z$ fault on an edge to the $Z$ fault acting on the primal qubit it connects to, and it also maps a $X$ fault on an edge to the $Z$ fault acting on the dual qubit it connects to. Note that chain map $f_2$ maps each $Z$ tensor to a primal qubit, and each $X$ tensor to a dual qubit. 
    Panel (a) depicts the action of chain map $f_1$. Panel (b) shows how the chain map $f_1$ acts on the gauge checks: it maps all weight-2 gauge checks to identity, keeping only the contributions from $g_{XZ}$ and $g_{ZX}$. Each $X$ gauge acting on a $Z$ tensor is mapped to the $Z$ component of a cluster state stabilizer acting on the primal qubit given by $Z$ tensor, and similarly for the $Z$ gauge acting on a $X$ tensor. Faults of the same Pauli type on the two half-edges of a contracted edge have identical images under $f_1$, so paired $XX$ and $ZZ$ faults map to the identity. We therefore depict one effective $X$-fault location and one effective $Z$-fault location per contracted edge.}
    \label{fig:chain_map}
\end{figure}

Since the twisted fault complex has a direct sum structure, we can decompose it as

\begin{equation}
\label{eq:twisted_fault_complex_decomposition}
\begin{tikzcd}[row sep=1.0em, column sep=2.0em]
F^1 \arrow[r, "\partial_{2,F}^\mathsf{T}"]  &
F^2 \arrow[r, "\partial_{3,F}^\mathsf{T}"] &
F^3  \\
& \bigoplus &\\
F_2  \arrow[r, "\partial_{2,F}"] &
F_1  \arrow[r, "\partial_{1,F}"] &
F_0 .
\end{tikzcd}
\end{equation}
The first is part of the cochain complex, while the second is part of the chain complex.

Then we take the dual complex of the first cochain chain and combine with the second chain complex to assemble the conventional fault complex
\begin{equation}
\begin{tikzcd}[row sep=1.2em, column sep=1.0em]
F_3 \arrow[r, "\partial_{3,F}"]&
F_2 \arrow[r, "\partial_{2,F}"] &
F_1 \arrow[r, "\partial_{1,F}"] & 
F_0.\\
{\scriptstyle \text{dual detectors}} &
{\scriptstyle \text{dual fault positions}} &
{\scriptstyle \text{primal fault positions}} &
{\scriptstyle \text{primal syndromes}}
\end{tikzcd}
\end{equation}

\subsection{Microscopic analysis of the chain map}

In this subsection, we analyze the chain map microscopically and show how  it provides a one-to-one mapping between the spacetime detector (which are called circuit detectors in circuit analysis and Pauli web in ZX calculus) and detectors in fault complexes.

Note that a CSS quantum circuit can be treated as a tensor formed by contracting individual $X$ and $Z$ tensors, now we apply the chain map $f_2$ that map each $Z$ tensor to a primal qubit in $F_1$, and map each $X$ tensor to a dual qubit in $F_2$. This gives us a bipartite graph state whose connectivity between primal and dual qubits is given by the connectivity between the $X$ and $Z$ tensors. The input, output, and measurement ports are gauge fixed and traced out, because fault complex describes the fault tolerant property within the bulk of a protocol. This gives us a bipartite graph $(V_X \cup V_Z, E)$, which corresponds to following cluster state

\begin{equation}\label{eq:cluster_complex}
\begin{tikzcd}[row sep=1.2em, column sep=1.0em]
F_2 \arrow[r, "\partial_{2,F}"] &
F_1  ,& \\
{\scriptstyle \text{dual fault positions}} &
{\scriptstyle \text{primal fault positions}} &
\end{tikzcd}
\end{equation}
where $\mathrm{dim}(F_2)=|V_X|, \mathrm{dim}(F_1)=|V_Z|$, and $\partial_{2,F}$ is the adjacency matrix describing the connectivity. This gives us cluster state stabilizers $\langle X_{\beta } Z_{\partial_{2,F} \beta },X_{\alpha } Z_{\partial_{2,F}^\mathsf{T} \alpha } \rangle $ for all $\alpha \in F_1, \beta\in F_2$. Here, $\alpha$ and $\beta$ identify the collections of $Z$ and $X$ tensors which fully specifies the generator labels of $g_{ZX}$ and $g_{XZ}$. The map $f_2$ annihilates the generator labels of type $g_{XX}$ and $g_{ZZ}$, while the labels of type $g_{XZ}$ and $g_{ZX}$ specify the $X$-components of cluster state stabilizers, respectively. The map $f_1$, acting on the Pauli supports of these generators, produces the corresponding $Z$-components. Together, these maps establish the correspondence between gauge checks of individual tensors and cluster state stabilizers
\begin{eqs}\label{eq:gauge_cluster_mapping}
g_{XZ}(v_x)
&\longleftrightarrow
X_{{v_x}}Z_{\partial_{2,F}{v_x}},
\quad v_x\in V_X,\\
g_{ZX}(v_z)
&\longleftrightarrow
X_{{v_z}}Z_{\partial_{2,F}^{\mathsf T}{v_z}},
\quad v_z\in V_Z,
\end{eqs}
where $v_x$ and $v_z$ label the tensors supporting $g_{XZ}$ and $g_{ZX}$, respectively.
This mapping can be understood as a compression of the redundant fault positions in the spacetime subsystem code description. Modulo the local $g_{ZZ}$ generators, any $Z$-fault pattern on the legs of a $Z$ tensor is equivalent to either the identity or a single $Z$ fault, according to whether its weight is even or odd. Similarly, modulo the local $g_{XX}$ generators, any $X$-fault pattern on the legs of an $X$ tensor reduces to either the identity or a single $X$ fault. Thus each $Z$ tensor contributes one effective $Z$-fault position, and each $X$ tensor contributes one effective $X$-fault position. Under $f_1$, these effective faults are represented by $Z$ faults on the associated primal and dual cluster state qubits, respectively.

\begin{lemma}\label{corollary: detector_support}
    The detectors in $\partial_{1,A}$ can be written as elements in $\mathcal{R}$, and they can also be written as elements in $\langle \mathcal{B} , \mathcal{M} \rangle $. A $X$-detector can be written as $\partial_{2,A} \mathbf{h}$ where $\mathbf{h}$ is a sum of two vectors    $\mathbf{h}=\mathbf{h}_X+\mathbf{h}_Z$ corresponding to the $g_{XX}$ and $g_{ZX}$, a $Z$-detector can be written as $\partial_{2,A} \mathbf{g}$ where $\mathbf{g}$ is a sum of two vectors $\mathbf{g}=\mathbf{g}_Z+\mathbf{g}_X$ corresponding to the $g_{ZZ}$ and $g_{XZ}$.  
\end{lemma}

\begin{corollary}
    A $X$-detector $D_{X,\mathbf{h}}=X(\partial_{2,A} \mathbf{h})$ specified by $\mathbf{h}=\mathbf{h}_X+ \mathbf{h}_Z \in A_2$ corresponds to the primal detector $X_{f_2(\mathbf{h}_Z)}$. Similarly, a $Z$-detector $D_{Z,\mathbf{g}}=Z(\partial_{2,A} \mathbf{g})$ specified by $\mathbf{g}=\mathbf{g}_Z+ \mathbf{g}_X \in A_2$ corresponds to the dual detector $X_{f_2(\mathbf{g}_X)}$. These correspondences follow from $f_2(\mathbf{h}_X)= f_2(\mathbf{g}_Z)=0$. The $X$- and $Z$-detectors give the boundary map $\partial_{1,\mathbf{F}}$.
\end{corollary}

Based on Lemma~\ref{corollary: detector_support}, we know the $X$- and $Z$-detectors are generated by the gauge group of individual tensors which are also required to commute with the entire gauge group, the detectors can be written as
\begin{eqs}
    D_{X,\mathbf{h}}=X(\partial_{2,A} \mathbf{h}) \equiv  W_{X}(a_{\mathbf{h}}),\quad D_{Z,\mathbf{g}}=Z(\partial_{2,A} \mathbf{g}) \equiv W_{Z}(b_{\mathbf{g}}),
\end{eqs}
where $\mathbf{h},\mathbf{g}$ are indices for gauge checks in $\mathcal{G}$, $a_{\mathbf{h}}, b_{\mathbf{g}} \subseteq E$ are subset of edges, and $W_{X}(a_\mathbf{h}), W_{Z}(b_\mathbf{g})$ represent the Pauli $X$ and $Z$ operators acting on the edges $a_{\mathbf{h}}, b_{\mathbf{g}}$. Note that $W_{X}(a_{\mathbf{h}}), W_{Z}(b_{\mathbf{g}})$ acts as $XX$ and $ZZ$ on the edges connecting to two tensors,. It does not have any single $X$ and $Z$ acting on the edge connecting to two tensors, because we have removed all the input, output, and measurement ports, and for contracted legs the single-qubit Pauli anticommute with the gauge group $\mathcal{B}$ hence it is impossible to appear.

The edge sets $a_\mathbf{h}, b_\mathbf{g}$ satisfy following incidence conditions with the vertices
\begin{eqs}\label{eq:incidence_conditions}
    \begin{cases}
    \mathcal{I} (v_x, a_\mathbf{h}) \mod 2=0, \quad \forall v_x \in V_X,\\
    \mathcal{I} (v_z, a_\mathbf{h}) \in \{0,\mathrm{deg}(v_z)\}, \quad \forall v_z \in V_Z,\\
    \mathcal{I} (v_z, b_\mathbf{g}) \mod 2=0, \quad \forall v_z \in V_Z,\\
    \mathcal{I} (v_x, b_\mathbf{g}) \in \{0,\mathrm{deg}(v_x)\}, \quad \forall v_x \in V_X,
    \end{cases}
\end{eqs}
where $\mathcal{I} (v, \tilde{E}) = |\{ e\in \tilde{E}: v\in \delta( e)\}|$ is the incidence number of a subset of edges $\tilde{E} \subseteq E$ on vertex $v$, here $\delta(e)$ denotes the endpoints of edge $e$ and $\mathrm{deg}(v )$ represents the degree of vertex $v$. These incidence constraints follow directly from the gauge groups of the individual $X$ and $Z$ tensors given in Eqs.~\eqref{eq:Rx} and \eqref{eq:Rz}. Those conditions are also named as Pauli web in the ZX literatures \cite{bombin2024unifying,rodatz2025fault}. These incidence conditions guarantees that $f_2, f_1,f_0$ are valid chain map, satisfying $f_1 \partial_{2,A}= \partial_{2,\mathbf{F}}f_2$ and $f_0 \partial_{1,A}= \partial_{1,\mathbf{F}}f_1$. 

We note that $a_h$ has either full or no incidence with any $Z$ tensor $v_z$, and $b_g$ has either full or no incidence with any $X$ tensor $v_x$, we can define $W_{X}, W_{Z}$ as a collection of $XX,ZZ$ operators defined on the center of $v_z, v_x$ and satisfy extra even incidence conditions with $v_x, v_z$ respectively. Above incidence conditions corresponds to the chain map $f_2$ removing all the entries corresponds to $g_{XX}$ and $g_{ZZ}$, and replace each $g_{ZX}$ and $g_{XZ}$ by a primal and dual fault position.

On the cluster state side we use $\alpha_{\mathbf{h}}\equiv f_2(\mathbf{h}), \beta_{\mathbf{g}} \equiv f_2(\mathbf{g})$ to denote the corresponding vertices that have full incidence with $a_\mathbf{h}$ and $ b_\mathbf{g}$
\begin{eqs}
    &\alpha_{\mathbf{h}}= f_2(\mathbf{h})=\{v_z \in V_Z: \mathcal{I}(v_z, a_\mathbf{h})=\mathrm{deg}(v_z)\} ,\\
    &\beta_{\mathbf{g}}=f_2(\mathbf{g})=\{v_x \in V_X: \mathcal{I}(v_x, b_\mathbf{g})=\mathrm{deg}(v_x)\} ,
\end{eqs}
such mapping is given by chain map $f_2$. Note that $\alpha_{\mathbf{h}}$ doesn't contain any vertices in $v_x$, and $\beta_{\mathbf{g}}$ doesn't contain any vertices in $v_z$ because of $f_2$ in Eq.~\eqref{eq:incidence_conditions} maps $g_{XX}, g_{ZZ}$ to zero.

After mapping to the cluster state, we apply single-qubit $X$ measurement to all the qubits. We find $W_{X}, W_{Z}$ gives us the primal and dual detectors which are cluster state stabilizers surviving after measurements
\begin{eqs}
    &W_{X}(a_\mathbf{h}) \longleftrightarrow  X_{f_2(\mathbf{h})} Z_{f_1(a_\mathbf{h})}=X_{f_2(\mathbf{h})} Z_{\partial_{2,F}^\mathsf{T} f_2(\mathbf{h})}=X_{\alpha_{\mathbf{h}}} Z_{\partial_{2,F}^\mathsf{T} \alpha_{\mathbf{h}}} =X_{\alpha_{\mathbf{h}}},\\
    &W_{Z}(b_\mathbf{g}) \longleftrightarrow X_{f_2(\mathbf{g})} Z_{f_1(b_\mathbf{g})}=X_{f_2(\mathbf{g})} Z_{\partial_{2,F} f_2(\mathbf{g})}=X_{\beta_{\mathbf{g}}} Z_{\partial_{2,F} \beta_{\mathbf{g}}} =X_{\beta_{\mathbf{g}}},
\end{eqs}
which are primal and dual detectors in $\partial_{1,F}$ and $\partial_{3,F}^\mathsf{T}$ acting on primal and dual qubits $\alpha_\mathbf{h}$ and $\beta_{\mathbf{g}}$ respectively. The primal and dual detectors are cluster state stabilizers centering on the primal qubits $\alpha_{\mathbf{h}}$ and dual qubits $\beta_{\mathbf{g}}$ respectively. The $Z_{\partial_{2,F}^\mathsf{T} \alpha_{\mathbf{h}}}=Z_{\partial_{2,F} \beta_{\mathbf{g}}}= \mathbb{I} $ because of the even incidence conditions in Eq.~\eqref{eq:incidence_conditions}, which give $\partial_{2,F}^\mathsf{T} \alpha_{\mathbf{h}} = \partial_{2,F} \beta_{\mathbf{g}}=0$. Equivalently, each $X$-/$Z$-detector in the ZX diagram maps to a primal/dual detector supported on the qubits corresponding to the $g_{ZX}$/$g_{XZ}$ gauge generators that constitute the original detector in the spacetime complex. Above equation also matches the chain map condition $f_1 \partial_{2,A} =\partial_{2,\mathbf{F}} f_2$. 

Hence $W_{X}(a_{\mathbf{h}})$ gives a primal detector acting on all the primal qubits which corresponds to the $Z$ tensors $\alpha_h$ on the support of $D_{X,\mathbf{h}}$, and $W_{Z}(b_{\mathbf{g}})$ gives a dual detector acting on all the dual qubits which corresponds to the $X$ tensors $\beta_{\mathbf{g}}$ on the support of $D_{Z,\mathbf{g}}$. The one-to-one correspondence between $X/Z$ detectors and primal/dual detectors automatically gives the identity $f_0$ between detector syndromes, and $f_0$ also satisfies the chain map condition $f_0 \partial_{1,A} =\partial_{1,\mathbf{F}} f_1$.

This gives a fault complex 

\begin{equation}
\begin{tikzcd}[row sep=1.2em, column sep=1.0em]
F_3 \arrow[r, "\partial_{3,F}"]&
F_2 \arrow[r, "\partial_{2,F}"] &
F_1 \arrow[r, "\partial_{1,F}"] & 
F_0,\\
{\scriptstyle \text{dual detectors}} &
{\scriptstyle \text{dual qubits}} &
{\scriptstyle \text{primal qubits}} &
{\scriptstyle \text{primal syndromes}}
\end{tikzcd}
\end{equation}
where the first and last boundary maps correspond to the dual and primal detector matrices $D_{\mathrm{dual}}^\mathsf{T}$ and $D_{\mathrm{primal}}$ respectively. 


This shows that the fault complex is an emergent object whose $X$ and $Z$ tensors form a bipartite graph and the spacetime complexes admit a direct sum structure. For general protocols including non-CSS protocols, the detector structure need not factorize into separate $X$- and $Z$-type components, and the spacetime complex provides the appropriate description. This enables us to map a CSS Clifford circuit or a CSS ZX diagram to a fault complex, simplifying the analysis and yielding a corresponding MBQC protocol.

\section{Discussion and outlook}

We established an algebraic formalism for spacetime fault tolerance in encoded Clifford protocols, unifying the analysis of different elements including quantum codes, gates, and measurements in the entire dynamical process. This formalism makes explicit the relationships among spacetime codes, circuit detectors, fault complexes, Gottesman's gadget framework, and the ZX calculus. Our criteria further relate spacetime fault distance and fault propagation to gadget correctness properties, connecting the algebraic description to the operational conditions used in standard gadget composition and threshold analyses~\cite{gottesman2024surviving}. This unified perspective suggests several directions worth further pursuing, which we outline below.

First, an important next step is to investigate conditions for positive noise thresholds in families of spacetime complexes. This calls for linking algebraic properties such as soundness or confinement to quantitative bounds on recovery failure probabilities and the residual noise passed between gadgets. Establishing such bounds under a certain noise model and showing that error suppression remains effective under repeated gadget composition would provide a route to useful spacetime threshold theorems.

Another valuable direction is to further develop algebraic constructions of low-overhead fault-tolerant protocols. 
Recent work on spacetime lifting gives fault complexes with almost-linear fault distance in the total spacetime cost, together with corresponding measurement-based realizations~\cite{xu2026framework}.
Advances in construction techniques for quantum LDPC codes~\cite{panteleev2022asymptotically,breuckmann2021balanced,dinur2024expansion} motivate extending these methods from static codes to spacetime complexes to further improve the achievable fault tolerance parameters. Relatedly, a useful goal is to characterize when algebraically defined spacetime complexes admit efficient circuit realizations and to develop systematic methods for constructing such realizations while preserving their fault tolerance properties.

From a practical perspective, our framework could guide the joint optimization of codes, circuits, and decoders under hardware and noise constraints. Building on detector-based methods, morphing circuits~\cite{mcewen2023relaxing,derks2025designing,shaw2025morphing,shaw2026optimising}, and ZX-based circuit transformations~\cite{rodatz2025fault,rusch2025completeness}, one could systematically explore more efficient implementations while keeping track of their fault-tolerance properties. It would also be valuable to investigate numerically how the algebraic structure of spacetime complexes influences logical error rates and decoding performance under circuit-level noise.

It is also likely fruitful to explore the connections between spacetime fault tolerance, many-body physics, and Hamiltonian complexity. A natural question is how the fault tolerance properties of a spacetime protocol are related to the energy barriers and thermal behaviors of an associated Hamiltonian~\cite{bravyi2011energy,finite_temperature_memory}. Such connections will shed light on the physical stability and computational properties of low-energy states, which are central objects of study in many-body physics and Hamiltonian complexity.  Also, drawing on ideas from self-correcting quantum memories~\cite{dennis2002topological,alicki2010thermal,balasubramanian2026passive}, one could further investigate what thermal stability implies for circuit-level fault tolerance and seek new approaches to protocol design and decoding.

To conclude, our theory provides an algebraic foundation for exploring the connections between spacetime fault tolerance, quantum coding theory, many-body physics, and Hamiltonian complexity, offering a broader perspective on the principles and physical realization of reliable quantum computation.

\section*{Acknowledgement }
Y.X.\ thanks Yujie Zhang and Yilun Li for hosting his research visits in Tokyo. Part of this work was done while Y.X.\ was attending the Fault Tolerant Quantum Technologies (FTQT) 2026 workshop hosted at the Benasque Science Center. Y.W is supported by a startup funding from SIMIS. Z.-W.L.\ is supported in part by NSFC under Grant No.~12475023, Dushi Program, and a startup funding from YMSC. 

\bibliography{biblo.bib}
\end{document}